\documentclass[journal,two column]{IEEEtran}

\usepackage{slashbox}
\usepackage{amsmath}
\usepackage{amssymb}
\usepackage{blkarray}
\usepackage{multirow}
\usepackage{mathtools}
\usepackage{amsthm}
\usepackage{graphicx}
\usepackage{times}
\usepackage[update,prepend]{epstopdf}
\usepackage[neverdecrease]{paralist}

\newcommand{\theset}[1]{\left\{#1\right\}}
\newcommand{\magn}[1]{\left|#1\right|}
\newcommand{\expect}[2]{\mathbf{E}_{#1}\left[#2\right]}
\newcommand{\prob}[1]{\mathbf{P}\left[#1\right]}

\newtheorem{theorem}{Theorem}

\newtheorem{proposition}[theorem]{Proposition}

\newtheorem{lemma}[theorem]{Lemma}
\newtheorem{algorithm}[theorem]{Algorithm}

\DeclareMathOperator*{\argmax}{arg\,max}
\DeclareMathOperator*{\argmin}{arg\,min}

\title{Efficient Strategy Computation in Zero-Sum Asymmetric Repeated Games}

\author{Lichun Li and Jeff S. Shamma
\thanks{Lichun Li is with the Coordinated Science Lab, University of Illinois at Urbana-Champaign, Urbana 61801, USA. {\small lichunli@illinois.edu}}
\thanks{J.S. Shamma is with the School of Electrical Engineering, King Abdullah University of Science and Technology (KAUST), Thuwal 23955, Saudi Arabia. {\small jeff.shamma@kaust.edu.sa}.}
\thanks{This research was supported by
 ARO project \#W911NF-09-1-0553 and the AFOSR/MURI project \#FA9550-10-1-0573 and by funding from King Abdullah University of Science and Technology (KAUST).}
}

\begin{document}

\maketitle

\begin{abstract}
Zero-sum asymmetric games model decision making scenarios involving two competing players who have different information about the game being played. A particular case is that of nested information, where one (informed) player has superior information over the other (uninformed) player. This paper considers the case of nested information in repeated zero-sum games and studies the computation of strategies for both the informed and uninformed players for finite-horizon and discounted infinite-horizon nested information games. For finite-horizon settings, we exploit that for both players, the security strategy, and also the opponent's corresponding best response depend only on the informed player's history of actions. Using this property, we refine the sequence form, and formulate an LP computation of player strategies that is linear in the size of the uninformed player's action set. For the infinite-horizon discounted game, we construct LP formulations to compute the approximated security strategies for both players, and provide a bound on the performance difference between the approximated security strategies and the security strategies. Finally, we illustrate the results on a network interdiction game between an informed system administrator and uniformed intruder.
\end{abstract}

\section{Introduction}
There are many competitive settings in which players have asymmetric information about the underlying state of the game. Examples include cyber security problems \cite{aziz2017jammer,zheng2012dynamic}, resource competitions in air transportation systems \cite{cruciol2015trajectory, li2017robust}, national defense \cite{kamble2015games,aumann1995repeated}, economic systems \cite{ouyang2015dynamic}, power networks \cite{st2004bayesian} and so on. In these systems, because of the non-cooperation between players, a player usually holds private information that is not shared with the other players, which causes the information asymmetry in games.

This paper focuses on two player zero-sum repeated games with asymmetric information, where one (informed) player knows the underlying state of the game while the other (uninformed) player only knows the prior distribution of the state. At the beginning of the game, the state of the game is initially selected, once and for all, at random according to the prior distribution. Repeated interactions means that players play over stages and can make observations about past play. Here we assume the case of full monitoring, i.e. a player observe the actions taken by both players, and perfect recall, i.e. a player records the history actions of both players. The one-stage payoff of the informed player, i.e. the loss of the informed player, is decided by the state of the game and the actions of both players. Notice that the payoff is not directly observable to both players, but the informed player can compute the payoff since he/she knows the state of the game and the actions of both players. The fact that the one-stage payoff is unavailable to the uninformed player happens in practice. For example, when jamming wireless sensor networks, the attacker may observe which channel the network uses (action of the informed player), but cannot measure the throughput in the channel (payoff of the informed player) \cite{aziz2017jammer,vadori2015jamming}.

We first study finite stage (N-stage) repeated asymmetric games. These games were studied in \cite{shimkin1996asymptotically} with a random payoff. Besides actions of players, the payoff is also a common information in \cite{shimkin1996asymptotically}. Based on these common information, the authors designed an asymptotically optimal strategy for the uninformed player such that the worst case relative loss of the uninformed player is either finite or increasing with rate $log(N)$ as $N$ goes to infinity. This paper adopts the game model in \cite{aumann1995repeated,zamir1992repeated} with fixed payoff function that is not available to the uninformed player. We are interested in computing security strategies of both players, i.e. the Nash equilibrium of the finite stage game. Notice that with the assumption of mutually absolute continuity in \cite{shimkin1996asymptotically}, the game model in \cite{shimkin1996asymptotically} is neither more nor less general than the one in this paper. Since the finite stage repeated asymmetric game has finite state set, finite action sets of both players, and finite stage, it can be expressed as a finite game tree. With perfect recall, the security strategies of both players can be computed by solving a linear program whose size is linear with respect to the size of the game tree, i.e. polynomial in the size of the action sets of both players, linear in the size of the state set, and exponential in the time horizon \cite{koller1996efficient}. Meanwhile, the other prior work showed that both players had a security strategy that is independent of the history actions of the uninformed player \cite{aumann1995repeated,de1996repeated,sorin2002first}. This paper shows that if a player uses a strategy that's independent of the uninformed player's history actions, then the opponent's best response is independent of the uninformed player's history actions, too. It implies that a player does not need to record the history actions of the uninformed player to compute the security strategy, which is a violation of the perfect recall. The challenge is how to develop a linear program without perfect recall as required by previous work \cite{koller1996efficient} to compute the security strategies for both players. For the informed player, to conquer the challenge, we started from the recursive formula of the game value, and develop an LP to compute the security strategy. For the uninformed player, to conquer the challenge, we cut the branches expanded by the uninformed player in the game tree, introduce the expected payoff at the leaf nodes, and construct an LP to compute the security strategies. Moreover, the sizes of the linear programs developed in this paper are only linear in the size of the action set of the uninformed player as compared with prior work that had polynomial dependence.

We then extend the time horizon to infinity, and study discounted repeated asymmetric games. Compared with finite stage games, discounted repeated asymmetric games has two main challenges. The first challenge is that it is necessary to find out fixed sized sufficient statistics for each player, since history based strategy requires a great amount of memory to record the history action as time horizon get long. The second challenge is that computing the game value and the security strategies of both players is non-convex \cite{gilpin2008solving,sandholm2010state}. We need to find an approximated security strategy for each player with guaranteed performance.

For the informed player, the first challenge has been addressed in the previous work \cite{aumann1995repeated}, and the sufficient statistics of the informed player is the posterior probability of the state of the game conditioned on the history action of the informed player, which is also called the belief. For the second challenge, we first use the game value of a finite stage ($N$-stage) discounted game, a truncated version of the infinite stage discounted game, to approximate the game value of the infinite stage discounted game, and then use an LP similar to the one in finite horizon games to compute an approximated security strategy based on the approximated game value. Such an approximated game value is shown to converge to the true game value exponentially fast, and the difference between the game value and the worst case payoff of the approximated security strategy has a finite upper bound which converges to $0$ exponentially fast with respect to $N$. The technique used in this section is adopted from our previous work \cite{li2015efficient} which focused on the informed player's approximated security strategy in discounted stochastic asymmetric games. This paper further studies the strategies of the uninformed player.

For the uninformed player, the belief can not be the uninformed player's sufficient statistics because of its dependency on the informed player's strategy. To figure out the sufficient statistics of the uninformed player, \cite{de1996repeated,sorin2002first} studied the dual game of the infinite stage discounted repeated game (primal game), and showed that the security strategy of the uninformed player in the dual game with a special initial regret (a real vector of the same size as the belief) is also a security strategy of the uninformed player in the primal game. Moreover, \cite{de1996repeated,sorin2002first} also showed that in the dual game, the sufficient statistics of the uninformed player is anti-discounted expected vector payoff realized so far conditioned on the history action of the informed player and the possible state of the game, which is called anti-discounted regret in this paper. Based on the results in the previous work, we first transform the primal game into a dual game with a special initial regret, and then compute an anti-discounted regret based security strategy for the uninformed player. The problem is what the special initial regret is. This paper shows that the special initial regret is the difference between $\mathbf{0}$ and the worst case payoff of the uninformed player's security strategy given every possible state of the game in the primal infinite stage discounted game. Computing the special initial regret is non-convex. Therefore, we use the worst case payoff of the uninformed player's security strategy given every possible state of the game in finite stage discounted game to approximate the one in the infinite stage discounted game, and construct a linear program to compute it. Given the approximated initial regret, computing the uninformed player's security strategy in a dual game is still non-convex. We first use the game value of a finite stage ($N$-stage) dual discounted game, a truncated version of the infinite horizon dual discounted game, to approximate the game value of the infinite stage dual discounted game, then construct an approximated security strategy of the uninformed player in the dual discounted game based on the approximated game value, and finally develop an LP to compute the approximated security strategy for the uninformed player in the infinite stage dual discounted game. Using the same technique as for the informed player and the relations between the game values of the primal and dual games, we show that in the infinite stage primal discounted game, the difference between the worst case performance of the uninformed player's approximated security strategy and the game value is always finite, and converges to $0$ exponentially fast with respect to $N$.% This difference is also called the worst case relative loss in \cite{shimkin1996asymptotically}. Considering total payoff, the uninformed player's asymptotical optimal strategy provided in \cite{shimkin1996asymptotically} guarantees that the worst case relative loss is either finite or increasing logarithmically as the time horizon goes to infinity. With discounted payoff, the uninformed player's approximated security strategy given in this paper guarantees the worst case relative loss is always finite, and decreases to $0$ %exponentially with respect to $N$.

This paper extended the results in our previous work \cite{li2016finite} from finite stage to infinite stage. For infinite stage games, we adopt the technique in our previous work \cite{li2015efficient} which studied the approximated security strategy of the informed player in discounted stochastic game, and extends the results from the informed player to the uninformed player. The remainder of this paper is organized as follows. Section II presents the main results for finite stage games. Section III discusses discounted infinite horizon games. Section IV illustrates the results on a network interdiction game. Finally, Section V presents some future work.

\section{Finite Stage Asymmetric Repeated Games}
\label{sec:finitehorizon}

\noindent{\textbf{Notation.}} Let $\mathbb{R}^n$ and $\mathbb{Z}^+$ denote $n$-dimensional real space and positive integers, respectively. Given a finite set $K$, its cardinality is denoted by $\magn{K}$, and $\Delta(K)$ is the set of probability distributions over $K$. The vectors $\mathbf{1}$ and $\mathbf{0}$ are appropriately dimensioned column vectors with all elements being $1$ and $0$, respectively. For $v(0),v(1),v(2),...$ a sequence of real numbers, we adopt the convention that $\sum_{t=1}^0 v(t)=0$, and $\prod_{t=1}^0 v(t)=1$. The supreme norm of a function $f: D\rightarrow \mathbb{R}$ is defined as $\|f\|_{\sup}=\sup_{x\in D}|f(x)|$, where $D$ is a non-empty set.

\subsection{Setup}
A two-player zero-sum asymmetric repeated game is specified by a five-tuple $(K,A,B,M,p_0)$, where
\begin{itemize}
  \item $K$ is a non-empty finite set, called the state set, the elements of which are called states.
  \item $A$ and $B$ are non-empty finite sets, called player 1 and 2's action sets, respectively.
  \item $M: K\times A \times B \rightarrow \mathbb{R}$ is the one-stage payoff function of player 1, or the one-stage penalty function of player 2. $M^k$ indicates the payoff matrix given state $k\in K$. The matrix element $M^k_{a,b}$, also denoted as $M(k,a,b)$, is the payoff given state $k\in K$, player 1's action $a\in A$, and player 2's action $b\in B$. The notation $M^k_{a,:}$ indicates the row vector payoff given state $k$ and player 1's action $a\in A$.
  \item $p_0\in \Delta(K)$ is the initial probability on $K$. We assume that $p_0^k>0$ for any $k\in K$.
\end{itemize}

A $N$-stage asymmetric repeated game is played as follows. Let $a_t,\ b_t$ denote the actions of player 1 and player 2 for stages $t\in \theset{1,2,...,N}$, respectively. At stage $t=1$, a state $k$ is chosen once and for all according to the probability distribution $p_0$, and communicated to player 1 only. Player 1 and 2 are called the informed and the uninformed player, respectively. Each player chooses his action independently, and the pair $(a_1,b_1)$ is observed by both players.  At stage $t=2$, both players again simultaneously choose their actions, and these are observable by both players. The payoff of player 1 and player 2 at stage $t$ is $M^k_{a_t,b_t}$ and $-M^k_{a_t,b_t}$, respectively. Since the sum of two players' payoffs is zero, this is a zero-sum game. The process is repeated for the remaining $t = 2,3,...,N$. These payoffs are not observed by player 2.

More formally, we will use the concept of behavior strategies. For any stage $t=1,\ldots,N$, the histories of player 1 and 2's actions prior to time $t$ are denoted by $h_t^A=\{a_1,\ldots,a_{t-1}\}$ and $h_t^B=\{b_1,\ldots,b_{t-1}\}$, respectively. For $t=1$, the null histories are denoted $h_1^A = h_1^B = \emptyset$. The corresponding \textit{set} of possible action sequences are denoted by $H^A_{t}=A^{t-1}$ and $H^B_t=B^{t-1}$.  A behavior strategy for player 1 is a collection of mappings $\sigma=(\sigma_t)_{t=1}^N$, where each $\sigma_t$ is a map from $K\times H^A_t \times H^B_t$ to $\Delta(A)$. Similarly, but taking into account the lack of information on the state $k\in K$, a behavior strategy for player 2 is a collection of mappings $\tau=(\tau_t)_{t=1}^N$, where $\tau_t$ is a map from $H^A_t \times H^B_t$ to $\Delta(B)$. Denote by $\Sigma$ and $\mathcal{T}$ the sets of behavior strategies of player 1 and 2, respectively.  The values $\sigma_t^a(k,h_t^A,h_t^B)$ for $a\in A$ and $\tau_t^b(h_t^A,h_t^B)$ for $b\in B$ denote the probabilities of playing $a$ and $b$ at stage $t$, respectively, given the histories $h_t^A\in H^A_t$ and $h_t^B\in H_t^B$, and realized state, $k\in K$.

Play proceeds as follows. As previously stated, at stage $t=1$, a state $k$ is chosen once and for all according to the probability distribution $p_0$. The action $a_1$ is a randomized outcome according to the behavior strategy distribution $\sigma_1(k,\emptyset,\emptyset)\in \Delta(A)$, and the action $b_1$ is a randomized outcome according to the behavior strategy distribution $\tau_1(\emptyset,\emptyset)$. At stage $t = 2,...,N$, the action $a_t$ is a randomized outcome according to the behavior strategy distribution $\sigma_t(k,h_t^A,h_t^B)\in \Delta(A)$, and the action $b_t$ is a randomized outcome according to the behavior strategy distribution $\tau_1(h_t^A,h_t^B)$, where we assume that these outcomes are conditionally independent given $h_t^A$ and $h_t^B$.

A triple $(p_0,\sigma,\tau)$ induces a probability distribution $P_{p_0,\sigma,\tau}$ on the set $\Omega=K\times (A\times B)^N$ of plays. Let
$\expect{p_0,\sigma,\tau}{\cdot}$ denote the corresponding expectation. The payoff with initial probability $p_0$ and strategies $\sigma$ and $\tau$ of the $N$-stage asymmetric information repeated game is defined as
\begin{align*}
  \gamma_N(p_0,\sigma,\tau)= \expect{p_0,\sigma,\tau}{\sum_{t=1}^{N}M(k,a_t,b_t)}.
\end{align*}

The $N$-stage game $\Gamma_T(p_0)$ is defined as the two-player zero-sum asymmetric repeated game equipped with initial distribution $p_0$, strategy spaces $\Sigma$ and $\mathcal{T}$, and payoff function $\gamma_N(p_0,\sigma,\tau)$. In game $\Gamma_N(p_0)$, the informed player seeks to \textit{maximize} the payoff $\gamma_N(p_0,\sigma,\tau)$, while the uninformed player seeks to \textit{minimize} it.

For the $N$-stage game $\Gamma_N(p_0)$, the security level $\underline{V}_N(p_0)$ of the informed player is defined as
$$\underline{V}_N(p_0)=\max_{\sigma\in\Sigma}\min_{\tau\in\mathcal{T}}\gamma_N(p_0,\sigma,\tau),$$
and the strategy $\sigma^*\in \Sigma$ which achieves the security level is called the security strategy of the informed player. Similarly, the security level $\overline{V}_N(p_0)$ of the uninformed player is defined as
$$\overline{V}_N(p_0)=\min_{\tau\in\mathcal{T}}\max_{\sigma\in\Sigma}\gamma_N(p_0,\sigma,\tau),$$
and the strategy $\tau^*\in \mathcal{T}$ which achieves the security level is called the security strategy of the uninformed player. When $\underline{V}_N(p_0)=\overline{V}_N(p_0)$, we say game $\Gamma_N(p_0)$ has a value, i.e. there exists a Nash equilibirum. Since the game $\Gamma_N(p_0)$ is a finite game, the game value always exists, and is denoted by $V_N(p_0)$ \cite{sorin2002first}.

%We will use the notation $a_s, h_s^A\in h_t^A$ to indicate that $a_s,h_s^A$ were the informed player's action and history action sequence, respectively, at stage $s<t$ in history action sequence $h_t^A$. Similarly, $b_s,h_s^B\in h_s^B$ indicates that $b_s,h_s^B$ where the uninformed player's action and history action sequence at stage $s<t$ in the history action sequence $h_t^B$.

\subsection{$H^B$ independent strategies}
\label{subsec: H B independent strategy}
A fundamental difference between a repeated asymmetric game and a one-shot asymmetric game is that in the repeated asymmetric game, the uninformed player can learn the system state from the informed player's actions. Indeed, the uninformed player's belief about the system state plays an important role for both players to make decisions \cite{aumann1995repeated}. Since only the informed player's actions are directly related to the system state, the uninformed player's history action sequence doesn't provide extra information about the system state given the informed player's history action sequence. Therefore, it is not surprised to see that given informed player's history action sequence, both players' security strategies are independent of the uninformed player's history action sequence \cite{sorin2002first}. Let's define an $H^B$ independent behavior strategy of player 1 as a collection of mappings $\bar{\sigma}=(\bar{\sigma}_t)_{t=1}^N$ where each $\bar{\sigma}_t$ is a map from $K\times H_t^A$ to $\Delta(A)$. Similarly, an $H^B$ independent behavior strategy of player 2 is a collection of mappings $\bar{\tau}=(\bar{\tau}_t)_{t=1}^N$ where $\bar{\tau}_t$ is a map from $H_t^A$ to $\Delta(B)$. Denote by $\bar{\Sigma}$ and $\bar{\mathcal{T}}$ the sets of $H^B$ independent behavior strategies of player 1 and 2. Clearly, $\bar{\Sigma}$ and $\bar{\mathcal{T}}$ are subsets of $\Sigma$ and $\mathcal{T}$, respectively.

\begin{proposition}[\cite{zamir1971relation,de1996repeated,sorin2002first}]
\label{prop:noP2}
Consider a two-player zero-sum $N$-stage asymmetric repeated game $\Gamma_N(p_0)$. Each player has a security strategy in game $\Gamma_N(p_0)$ that
is independent of player 2's history action sequence, i.e.
\begin{align*}
  \max_{\sigma\in \Sigma}\min_{\tau\in\mathcal{T}} \gamma_N(p_0,\sigma,\tau)=&\max_{\bar{\sigma}\in \Sigma}\min_{\tau\in\mathcal{T}} \gamma(p_0,\bar{\sigma},\tau)\\
  \min_{\tau\in\mathcal{T}}\max_{\sigma\in \Sigma} \gamma_N(p_0,\sigma,\tau)=&\min_{\bar{\tau}\in\bar{\mathcal{T}}} \max_{\sigma\in \Sigma} \gamma_N(p_0,\sigma,\bar{\tau}).
\end{align*}
\end{proposition}

If one player's behavior strategy is independent of the uninformed player's history action sequence, then the other player's best response to the $H^B$ independent strategy is independent of the uninformed player's history action sequence, too.
\begin{proposition}
\label{prop:noP2also}
Consider a two-player zero-sum $N$-stage asymmetric repeated game $\Gamma_N(p_0)$. For any $\bar{\sigma}\in \bar{\Sigma}$, and any $\bar{\tau}\in \bar{\mathcal{T}}$,
\begin{align}
  \min_{\tau\in \mathcal{T}}\gamma_N(p_0,\bar{\sigma},\tau)=&\min_{\bar{\tau}\in \bar{\mathcal{T}}} \gamma_N(p_0,\bar{\sigma},\bar{\tau}) \label{eq: noP2-1}\\
  \max_{\sigma\in \Sigma}\gamma_N(p_0,\sigma,\bar{\tau})=&\max_{\bar{\sigma}\in\bar{\Sigma}}\gamma_N(p_0,\bar{\sigma},\bar{\tau}). \label{eq: noP2-2}
\end{align}
\end{proposition}
\begin{proof}
Since $\bar{\mathcal{T}}\in \mathcal{T}$, we have $\min_{\tau\in \mathcal{T}}\gamma_N(p_0,\bar{\sigma},\tau)\leq\min_{\bar{\tau}\in \bar{\mathcal{T}}} \gamma_N(p_0,\bar{\sigma},\bar{\tau})$. Meanwhile, for any $\tau\in \mathcal{T}$, we can design $\bar{\tau}_t^{b_t}(h_t^A)=\sum_{h_t^B\in H_t^B}\prod_{s=1}^t\tau_s^{b_s}(h_s^A,h_s^B)$ for all $t=1,\ldots,N,$ such that $\gamma_N(p_0,\bar{\sigma},\tau)=\gamma_N(p_0,\bar{\sigma},\bar{\tau})$. Hence, we have $\min_{\tau\in \mathcal{T}}\gamma_N(p_0,\bar{\sigma},\tau) \geq \min_{\bar{\tau}\in \bar{\mathcal{T}}} \gamma_N(p_0,\bar{\sigma},\bar{\tau}).$ Therefore, equation (\ref{eq: noP2-1}) is shown.

Similarly, $\bar{\Sigma}\in \Sigma$ implies that $\max_{\sigma\in \Sigma}\gamma_N(p_0,\sigma,\bar{\tau}) \geq \max_{\bar{\sigma}\in\bar{\Sigma}}\gamma_N(p_0,\bar{\sigma},\bar{\tau})$. Meanwhile, for any $\sigma\in\Sigma$, we can design $\bar{\sigma}_t^{a_t}(k,h_t^A)
=\frac{\sum_{h_t^B\in H_t^B} \prod_{s=1}^t \sigma_s^{a_s}(k,h_s^A,h_s^B) \prod_{s=1}^{t-1}\bar{\tau}_s^{b_s}(h_s^A)}{\sum_{h_{t-1}^B\in H_{t-1}^B} \prod_{s=1}^{t-1} \sigma_s^{a_s}(k,h_s^A,h_s^B) \prod_{s=1}^{t-2}\bar{\tau}_s^{b_s}(h_s^A)}$ for all $t=1,\ldots,N$, such that $\gamma_N(p_0,\sigma,\bar{\tau})=\gamma_N(p_0,\bar{\sigma},\bar{\tau})$, which implies that $\max_{\sigma\in \Sigma}\gamma_N(p_0,\sigma,\bar{\tau}) \leq \max_{\bar{\sigma}\in\bar{\Sigma}}\gamma_N(p_0,\bar{\sigma},\bar{\tau})$. Therefore, equation (\ref{eq: noP2-2}) is shown.
\end{proof}
%%%%%%%% proof may need, or reference is needed
Proposition \ref{prop:noP2} and \ref{prop:noP2also} imply that when computing players' security strategies, we can ignore the uninformed player's history action sequence, which greatly reduces the number of both players' information sets in the extensive game tree, and hence reduces the computational complexity of the security strategies.

\subsection{LP formulations of security strategies}
A $N$-stage asymmetric information repeated game, as a finite game, can always be expressed as a finite extensive game tree\cite{kreps1982sequential}. Assuming perfect recall, i.e. each player can record all history actions of both players, we can use sequence form to construct a linear program to compute the security strategy. Roughly speaking, in two-player zero-sum games, with sequence form, the total payoffs at the leaf nodes of the game tree are provided first, then the probability of the sequence from the root node to the leaf node is characterized in the form of every player's realization plan, i.e. a player's strategy production, and finally an LP formulation can be derived to compute the security strategies of two players based on an LP's dual problem. Moreover, the size of the linear program is linear in the size of the game tree, and hence polynomial in the size of the uninformed player's action set \cite{koller1996efficient}. In our case, the analysis in subsection \ref{subsec: H B independent strategy} indicates that both players can ignore the uninformed player's history action sequence when making decisions. In other words, the uninformed player can forget what he did before, which violates the perfect recall assumption in the sequence form. Here, we will adopt the realization plan in the sequence form, and take advantage of the $H^B$ independent strategies to develop LP formulations with reduced computational complexity to compute the $H^B$ independent security strategies.

As in the sequence form, we define the \emph{realization plan} $q_t(h_t^A;k)$ of the informed player's history action sequence $h_t^A$ given state $k$ at stage $t$ as
\begin{align}
  q_t(h_t^A;k)=\prod_{s=1}^{t-1}\bar{\sigma}^{a_s}_s(k,h_s^A), \label{eq: q}
\end{align}
where $a_s$ and $h_s^A$ are the informed player's action and history action sequence at stage $s$ in the history action sequence $h_t^A$, denoted by $a_s,h_s^A\in h_t^A$. Therefore, the realization plan $q$ satisfy the following constraints:
\begin{align}
 q_1(h_1^A;k)&=1,&&\forall k\in K, \label{eq: q constraint 1}\\
 \sum_{a_t\in A}q_{t+1}((h_t^A,a_t);k)&=q_{t}(h_t^A;k),&&\forall k\in K, h_{t}^A\in H_{t}^A,\nonumber\\
  &&& \forall t=1,\ldots,N,  \label{eq: q constraint 2}\\
  q_t(h_t^A;k)&\geq 0, &&\forall k\in K, h_{t}^A\in H_{t}^A,\nonumber \\
&&& \forall t=2,\ldots,N+1, \label{eq: q constraint 3}
\end{align}
where $(h_t^A,a_t)$ indicates concatenation. A realization plan of the informed player is a collection of the informed player's realization plans $q=(q_t)_{t=1}^{N+1}$ at all stages. Indeed, the realization plan $q_t(h_t^A;k)$ is the conditional probability $\prob{h_t^A|k}$. The set of realization plans of the informed player is denoted by $Q$, including all properly dimensioned real vectors satisfying equation (\ref{eq: q constraint 1}-\ref{eq: q constraint 3}).

A very important difference between a one-shot game and a repeated game is that the uninformed player can learn the system state from the informed player's history actions. The informed player can characterize his revelation of information by the posterior probability $\prob{k|h_t^A}$, which is also called the belief state of player 2. Let $p_t\in \Delta(K)$ denote the posterior probability over the system state $k\in K$ at stage $t$ given $h_t^A$, i.e. $p_t^k(h_t^A)=\prob{k|h_t^A}$. The belief state $p_{t+1}$ at stage $t+1$ can be computed recursively as a function of $p_t$, the informed player's strategy $x_t^k=\bar{\sigma}_t(k,h_t^A)$, and the informed player's realized action $a_t$ based on the Bayesian law. Therefore, we have
\begin{align}
  p_{t+1}^k(h_{t+1}^A)=\pi(p_t,x_t,a_t)=\frac{p_t^k(h_t^A)x_t^k(a_t)}{\bar{x}_{p_t,x_t,(a_t)}}, \label{eq: belief state}
\end{align}
with $p_1=p$ in game $\Gamma_N(p)$. Here, $x_t^k(a_t)=\sigma_t^{a_t}(k,h_t^A)$, and $\bar{x}_{p_t,x_t}(a_t)=\sum_{k\in K}p_t^k(h_t^A)x_t^k(a_t)$. The variable $\bar{x}$ can be seen as the weighted average of $x_t$. Based on the belief state, the value function $V_N(p)$ satisfies a backward recursive equation which is similar to the Bellman's equation \cite{aumann1995repeated,sorin2002first}.
\begin{align}
  V_t(p)=&\max_{x\in \Delta(A)^{|K|}}\min_{y\in \Delta(B)}\sum_{k\in K}p^k{x^k}^TM^ky \nonumber\\
  &+\sum_{a_1\in A}\bar{x}_{p,x}(a_1)V_{t-1}(\pi(p,x,a_1)). \label{eq: recursive formula, finite stage}
\end{align}

Based on the realization plan $q$ and the backward recursive formula (\ref{eq: recursive formula, finite stage}), we construct a linear program to compute the security strategy for the informed player.
\begin{theorem}
\label{theorem: LP formula for informed one in T-stage game}
Consider a two-player zero-sum $N$-stage asymmetric repeated game $\Gamma_N(p)$. The game value $V_N(p)$ satisfies
\begin{align}
V_N(p)
 =& \max_{q,\ell\in Q , L} \sum_{t=1}^N\sum_{h_t^A\in H_t^A} \ell_{h_t^A}\label{eq: T stage, informed, LP formula}\\
 s.t. & \  \sum_{k\in K,a\in A} p^k q_{t+1}((h_t^A,a);k)M^k_{a,:} \geq \ell_{h_t^A} \mathbf{1}^T , \nonumber\\
 & \forall t=1,\cdots,N, \forall h_t^A\in H_t^A. \label{eq: informed player's main constraint}
 \end{align}
 where $Q$ is a set including all properly dimensioned real vectors satisfying (\ref{eq: q constraint 1}-\ref{eq: q constraint 3}), $L$ is a properly dimensioned real space, and $(h_t^A,a)$ indicates concatenation. The informed player's security strategy $\bar{\sigma}^*$ is
 \begin{align}
 \bar{\sigma}^{a*}_t(k,h_t^A)=q_{t+1}^*((h_t^A,a);k)/q_t^*(h_t^A;k), \forall a\in A. \label{eq: informed player's security strategy, T stage}
 \end{align}
\end{theorem}
\begin{proof}
By the duality theorem \footnote{Consider a primal LP problem and the corresponding dual LP problem. If solutions to both problem exists, then optimal feasible solutions to both problems exist, and the optimal values of the two problems are equal. \cite{dantzig2016linear}}, it is easy to see that equation (\ref{eq: T stage, informed, LP  formula}-\ref{eq: informed player's main constraint}) is true for $N=1$. Let's assume that $V_{t-1}(p)$ satisfies (\ref{eq: T stage, informed, LP  formula}-\ref{eq: informed player's main constraint}) for all $t=2,\ldots.$ According to Lemma III.1 of \cite{li2014lp}, we have $\bar{x}_{p,x}(a_1)V_{t-1}(\pi(p,x,a_1))=V_{t-1}(\bar{x}_{p,x}\pi(p,x,a_1)).$ Therefore, the second term of (\ref{eq: recursive formula, finite stage}) satisfies
\begin{align*}
&\sum_{a_1\in A}\bar{x}_{p,x}(a_1)V_{t-1}(\pi(p,x,a_1))\\
=&\sum_{a_1\in A}\max_{q,\ell_{a_1}\in Q,L_{a_1}}\sum_{s'=1}^{t-1}\sum_{h_{s'}^A\in H_{s'}^A}\ell_{(a_1,h_{s'}^A)}\\
s.t. & \sum_{k\in K}p^k x^k(a_1)q_{s'+1}((h_{s'}^A,a);k)\geq \ell_{(a_1,h_{s'}^A)}\mathbf{1}^T, \\
&\forall s'=1,\ldots,t-1, h_{s'}^A\in H_{s'}^A.
\end{align*}
Let $s=s'+1$ and $h_s^A=(a_1,h_{s'}^A)$. we have
\begin{align*}
&\sum_{a_1\in A}\bar{x}_{p,x}(a_1)V_{t-1}(\pi(p,x,a_1))\\
=&\max_{q,\ell\in Q,L}\sum_{s=2}^{t}\sum_{h_{s}^A\in H_{s}^A}\ell_{h_s^A}\\
s.t. & \sum_{k\in K}p^k q_{s+1}((h_{s}^A,a);k)\geq \ell_{h_s^A}\mathbf{1}^T, \\
&\forall s=2,\ldots,t, h_{s}^A\in H_{s}^A.
\end{align*}

By the duality theorem, it is easy to verify that
\begin{align*}
\min_{y\in \Delta(B)}\sum_{k\in K}p^k{x^k}^TM^ky=&\max_{\ell_{h_1^A}\in \mathbb{R}} \ell_{h_t^A} \\
s.t. & \sum_{k\in K}p^k{x^k}^TM^k \geq \ell_{h_1^A}\mathbf{1}^T.
\end{align*}
According to equation (\ref{eq: recursive formula, finite stage}), and with the fact that $x^k(a_1)=q_2(a_1;k)$, we show that equation (\ref{eq: T stage, informed, LP  formula}-\ref{eq: informed player's main constraint}) still holds for $V_t(p)$ for $t=2,\ldots.$

Once we get the optimal solution $q^*$, according to (\ref{eq: q}), the security strategy of the informed player can be computed according to (\ref{eq: informed player's security strategy, T stage}).
\end{proof}

Our LP formulation of informed player's security strategy has its size linear in the size of the state set and the size of uninformed player's action set, polynomial in the size of informed player's action set, and exponential in time horizon. Let's first analyze the variable size. Variable $q$ consists of $(q_t)_{t=1}^{N+1}$, where $q_t$ is of size $|K|\times|H_t^A|=|K|\times|A^{t-1}|$, and hence $q$ consists of $|K|(1+|A|+\cdots+|A|^N)=O(|K||A|^{N+1})$ scalars. Variable $\ell$ consists of $(1+|A|+\cdots+|A|^{N-1})=O(|A|^N)$ scalars. In all, we see that the LP formulation has $O(|K||A|^{N+1})$ scalar variables. Next, let's take a look at the constraint size. Constraint (\ref{eq: q constraint 1}) includes $|K|$ equations. Constraint (\ref{eq: q constraint 2}) includes $|K|(1+|A|+\cdots+|A|^{N-1})=O(|K||A|^N)$ equations. Constraint (\ref{eq: q constraint 3}) includes $|K|(1+|A|+\cdots+|A|^{N})=O(|K||A|^{N+1})$ equations. Constraint (\ref{eq: informed player's main constraint}) includes $|B|(1+|A|+\cdots+|A|^{N-1})=O(|B||A|^{N})$ equations. In all, there are $O((|K|+|B|)|A|^{N+1})$ equations. Therefore, the size of the LP formulation to compute the informed player's security strategy is linear in $|K|$ and $|B|$, polynomial in $|A|$, and exponential in $N$.

Next, let's take a look at the uninformed player's security strategy. Define the conditional expected total payoff $u(\bar{\tau};k,h_{N+1}^A)$ given uninformed player's strategy $\bar{\tau}\in \bar{\mathcal{T}}$, state $k\in K$, and informed player's history action sequence $h_{N+1}^A\in H_{N+1}^A$ as
\begin{align}
  u(\bar{\tau};k,h_{N+1}^A)=\expect{\bar{\tau}}{\sum_{t=1}^N M(k,a_t,b_t)|k,h_{N+1}^A}. \label{eq: realized vector payoff}
\end{align}
It is easy to show that
\begin{align}
  u(\bar{\tau};k,h_{N+1}^A)=\sum_{t=1}^N M^k_{a_t,:}y_{h_t^A}, \label{eq: realized vector payoff matrix form}
\end{align}
where $y_{h_t^A}=\bar{\tau}_t(h_t^A)$, and $a_t,h_t^A\in h_{N+1}^A$. We notice that $u(\bar{\tau};k,h_{N+1}^A)$ is a linear function of $\bar{\tau}$, or in other words, $y$.

\begin{theorem} \label{theorem: LP formulation of uninformed player in t stage games}
Consider a two-player zero-sum $N$-stage asymmetric repeated game $\Gamma_N(p)$. The game value $V_N(p)$ satisfies
\begin{align}
  V_N(p)
=&\min_{y \in Y,\ell\in \mathbb{R}^{|K|}} p^T \ell \label{eq: LP for uninformed player}\\
s.t. &\ u(y;k,:) \leq\ell^k\mathbf{1}, &&\forall k\in K, \label{eq: uninformed player's main constraint}\\
&\mathbf{1}^Ty_{h_t^A}=1, &&\forall h_t^A\in H_t^A, \forall t=1,\ldots,N,\label{eq: y constraint 1}\\
&y_{h_t^A}\geq \mathbf{0}, &&\forall h_t^A\in H_t^A, \forall t=1,\ldots,N. \label{eq: y constraint 2}
\end{align}
where $Y$ is a properly dimensioned real space, and $u(y;k,:)$ is a $|H_{N+1}^A|$ dimensional column vector whose element is $u(y;k,h_{N+1}^A)$, a linear function of $y$ satisfying equation (\ref{eq: realized vector payoff matrix   form}). The uninformed player's security strategy $\bar{\tau}^*(h_t^A)$ is $y^*(h_t^A)$.
\end{theorem}
\begin{proof}
Let's define $\nu_N^k(\bar{\tau})=\max_{\bar{\sigma}(k)\in \bar{\Sigma}(k)} \expect{\bar{\sigma},\bar{\tau}}{\sum_{t=1}^N M^k_{a_t,b_t}|k},$ where $\bar{\sigma}(k)$ indicates the informed player's $H^B$ independent behavior strategy given the system state $k\in K$ and $\bar{\Sigma}(k)$ is the corresponding set including all possible $\bar{\sigma}(k)$. We have
\begin{align*}
&\nu_N^k(\bar{\tau})=\max_{\bar{\sigma}(k)\in \bar{\Sigma}(k)} \sum_{h_{N+1}^A\in H_{N+1}^A}\prob{h_{N+1}^A|k} u(\bar{\tau};k,h_{N+1}^A)\\
=&\max_{q_{N+1}(:,k)\in \Delta(H_{N+1}^A)} \sum_{h_{N+1}^A\in H_{N+1}^A}q_{N+1}(h_{N+1};k)u(\bar{\tau};k,h_{N+1}^A).
\end{align*}
According to the duality theorem, we have
\begin{align}
  \nu_N^k(\bar{\tau})=&\min_{\ell^k\in \mathbb{R}}\ell^k \label{eq: 4}\\
  s.t. &\ u(\bar{\tau};k,:)\leq \ell^k \mathbf{1}. \label{eq: 5}
\end{align}

The game value $V_N(p)$ satisfies
\begin{align*}
V_N(p)=&\min_{\bar{\tau}\in\bar{\mathcal{T}}}\sum_{k\in K} p^k \mu_N^k(\bar{\tau})\\
=& \min_{y\in Y, \ell\in \mathbb{R}^{|K|}} \sum_{k\in K} p^k \ell^k\\
s.t. & u(y;k,:) \leq \mathbf{1}\ell^k, \forall k\in K.
\end{align*}
\end{proof}

The LP formulation of the uninformed player's security strategy has its size linear in the size of the state set and his own action set, polynomial in the size of the informed player's action set, and exponential in time horizon. We first analyze the variable size. Variable $y$ includes $(y_t)_{t=1}^N$, where $y_t$ is of size $|B||A^{t-1}|$, and hence $y$ has $|B|(1+|A|+\cdots+|A|^{N-1})=O(|B||A|^N)$ scalar variables. Variable $\ell$ is of size $|K|$. In all, the variable size is in the order of $|B||A|^N+|K|$. We then study the constraint size. Constraint (\ref{eq: y constraint 1}) consists of $(1+|A|+\cdots+|A|^{N-1})=O(|A|^N)$ equations. Constraint (\ref{eq: y constraint 2}) consists of $|B|(1+|A|+\cdots+|A|^{N-1})=O(|B||A|^N)$ equations. Constraint (\ref{eq: uninformed player's main constraint}) consists of $|A|^N|K|$ equations. In all, the constraint size is of order $O((|B|+|K|+1)|A|^N)$. Therefore, the size of the LP formulation to compute the uninformed player's security strategy is linear in $|K|$ and $|B|$, polynomial in $|A|$, and exponential in $N$.

\section{$\lambda$-Discounted Asymmetric Repeated Games}
In finite-stage asymmetric information repeated games, the security strategies of the players depend on the informed player's history actions. As the time horizon gets long, players need a large amount of memory to record the history actions. Since the horizon of a $\lambda$-discounted asymmetric repeated game is infinite, it is necessary for players to find fixed-sized sufficient statistics for decision making. After figuring out the fixed-sized sufficient statistics, we find that players' security strategies in $\lambda$-discounted asymmetric repeated game are still hard to compute, and hence approximated security strategies with guaranteed performance are provided. This section talks about the sufficient statistics and the approximated security strategies player by player.

\subsection{Setup}
A two-player zero-sum $\lambda$-discounted asymmetric repeated game is specified by the same five-tuple $(K,A,B,M,p_0)$ and played in the same way as described in the two-player zero-sum $N$-stage asymmetric repeated game. The payoff of player 1 at stage $t$ is $\lambda(1-\lambda)^{t-1}M(k,a_t,b_t)$ for some $\lambda\in (0,1)$, and the game is played for infinite horizon. The payoff of the $\lambda$-discounted asymmetric repeated game with initial probability $p_0$ and strategies $\sigma$ and $\tau$ is defined as
\begin{align*}
  \gamma_{\lambda}(p_0,\sigma,\tau)=E_{p_0,\sigma,\tau}\left(\sum_{t=1}^\infty \lambda(1-\lambda)^{t-1} M(k,a_t,b_t)\right).
\end{align*}

The $\lambda$-discounted game $\Gamma_\lambda(p_0)$ is defined as a two-player zero-sum asymmetric repeated game equipped with initial distribution $p_0$, strategy spaces $\Sigma$ and $\mathcal{T}$, and payoff function $\gamma_\lambda(p_0,\sigma,\tau)$. The security strategies $\sigma^*$ and $\tau^*$, and security levels $\underline{V}_\lambda(p_0)$ and $\overline{V}_\lambda(p_0)$ are defined in the same way as in $N$-stage game in Section \ref{sec:finitehorizon} for player 1 and 2, respectively. Since $\gamma_\lambda(p_0,\sigma,\tau)$ is bilinear over $\sigma$ and $\tau$, $\Gamma_\lambda(p_0)$ has a value $V_\lambda(p_0)$ according to Sion's minimax Theorem, i.e. $V_\lambda(p_0)=\underline{V}_\lambda(p_0)=\overline{V}_\lambda(p_0)$ \cite{sorin2002first}.

\subsection{The informed player}

\subsubsection{The informed player's security strategy}
The belief state $p_t$ in (\ref{eq: belief state}) plays an important role in decision making of the informed player. Indeed, in a $\lambda$-discounted asymmetric repeated game $\Gamma_\lambda(p)$, the belief state $p_t$ is sufficient statistics of the informed player.
\begin{proposition}[\cite{sorin2002first}]
  \label{proposition: recursive formula in discounted games}
Consider a two-player zero-sum $\lambda$-discounted asymmetric repeated game $\Gamma_\lambda(p)$. The game value $V_\lambda(p)$ satisfies
\begin{align}
  V_\lambda(p)= &\max_{x\in \Delta(A)^{|K|}} \min_{y\in\Delta(B)} \nonumber\\
&\left(\lambda \sum_{k\in K} p^k{x^{k}}^TM^ky+(1-\lambda)\mathbf{T}_{p,x}(V_\lambda) \right), \label{eq: dynamic programming}
\end{align}
where
\begin{align}
  \mathbf{T}_{p,x}(V_\lambda)=\sum_{a\in A}\bar{x}_{p,x}^aV_\lambda(\pi(p,x,a)). \label{eq: T}
\end{align}
Moreover, the informed player has a security strategy that depends only on the belief state $p_t$ at each stage $t$, and is independent of the uninformed player's history action sequence.
\end{proposition}

First of all, Proposition \ref{proposition: recursive formula in discounted games} points out that the informed player's security strategy is independent of the uninformed player's history action, just as what it is in $N$-stage game. Following the same steps, we can show that the uninformed player's best response to an $H^B$ independent strategy is also $H^B$ independent. Second, Proposition \ref{proposition: recursive formula in discounted games} provides the sufficient statistics $p_t$ of the informed player. So the informed player only needs to record $p_t\in \Delta(K)$ instead of all of his own history actions. Finally, given the belief state $p_t$, Proposition \ref{proposition: recursive formula in discounted games} gives a Bellman-like equation (\ref{eq: dynamic programming}) to compute the informed player's security strategy.

Unfortunately, computing the value $V_\lambda(p)$ and the informed player's corresponding security strategy $\sigma^*$ is non-convex \cite{gilpin2008solving,sandholm2010state}. Therefore, we need to find an approximated security strategy that is easy to compute, and has some performance guarantee.

\subsubsection{The informed player's approximated security strategy}
One way to approximate the security strategy is to approximate the game value $V_\lambda(p)$ first, and then compute the security strategy based on the approximated game value. Here, we will use the game value $V_{\lambda,N}(p)$ of a $\lambda$-discounted $N$-stage asymmetric repeated game $\Gamma_{\lambda,N}(p)$ to approximate the game value $V_\lambda(p)$.

A $\lambda$-discounted $N$-stage repeated asymmetric game $\Gamma_{\lambda,N}(p_0)$ is a truncated version of a $\lambda$-discounted asymmetric repeated game $\Gamma_\lambda(p_0)$ with time horizon $N$. To be more specific, a $\lambda$-discounted $N$-stage asymmetric repeated game $\Gamma_{\lambda,N}(p_0)$ is specified by the same five-tuple $(K,A,B,M,p_0)$ and played in the same way as in a $\lambda$-discounted asymmetric repeated game $\Gamma_\lambda(p_0)$. The one-stage payoff is the same as in $\Gamma_\lambda(p_0)$, i.e. $\lambda(1-\lambda)^{t-1}M(k,a_t,b_t)$. The only difference between a $\lambda$-discounted $N$-stage repeated asymmetric game $\Gamma_{\lambda,N}(p_0)$ and a $\lambda$-discounted repeated asymmetric game $\Gamma_\lambda(p_0)$ is that the game is played for $N$ stages in $\Gamma_{\lambda,N}(p_0)$, and infinite stages in $\Gamma_\lambda(p_0)$. Therefore, the payoff of game $\Gamma_{\lambda,N}(p_0)$ is
\begin{align*}
    \gamma_{\lambda,N}(p_0,\sigma,\tau)=E_{p_0,\sigma,\tau}\left(\sum_{t=1}^N \lambda(1-\lambda)^{t-1} M(k,a_t,b_t)\right).
\end{align*}
A $\lambda$-discounted $N$-stage repeated asymmetric game $\Gamma_{\lambda,N}(p_0)$ is defined as a two-player zero-sum repeated asymmetric game equipped with initial probability $p_0$, strategy spaces $\Sigma$ and $\mathcal{T}$, and payoff function $\gamma_{\lambda,N}(p_0,\sigma,\tau)$.

Following the standard arguments as in the proof of Proposition \ref{proposition: recursive formula in discounted games}, we see that the game value $V_{\lambda,N+1}(p)$ of the $\lambda$-discounted $N$-stage game $\Gamma_{\lambda,N+1}(p)$ satisfies the recursive formula as below.
\begin{align}
  V_{\lambda,N+1}(p)
=&\max_{x\in\Delta(A)^{|K|}}\min_{y\in\Delta(B)}\left(\lambda \sum_{k\in K} p^k{x^k}^TM^ky\right.\nonumber\\
& \left.+(1-\lambda)\mathbf{T}_{p,x}(V_{\lambda,N})\right), \label{eq: value iteration}
\end{align}
with $V_{\lambda,0}(p)\equiv 0$.

Before we go ahead to provide the approximated security strategy based on this approximated game value, we are interested in how good the approximated game value is, and how fast it converges to the real game value. To this purpose, we define an operator $\mathbf{F}_x$ as
\begin{align}
\mathbf{F}_{x}^V(p)=&\min_{y\in \Delta(B)}\{\lambda \sum_{k\in K} p^k{x^k}^TM^ky +(1-\lambda)\mathbf{T}_{p,x}(V)\}.
\label{eq: F}
\end{align}
It's clear that $V_\lambda(p)=\max_{x\in\Delta(A)^{|K|}} \mathbf{F}_{x}^{V_\lambda}(p)$, and $V_{\lambda,N+1}(p)=\max_{x\in\Delta(A)^{|K|}} \mathbf{F}_{x}^{V_{\lambda,N}}(p)$. The operator $\mathbf{F}_x$ is actually a contraction mapping.

\begin{lemma}
\label{lemma: contractor F}
Let $\mathcal{V}$ be the set of functions mapping from $\Delta(K)$ to $\mathbb{R}$. Given any $x\in\Delta(A)^{|K|}$ and $\lambda\in (0,1)$, the operator $\mathbf{F}_x:\mathcal{V}\rightarrow\mathcal{V}$ defined in (\ref{eq: F}) is a contraction mapping with contraction constant $1-\lambda,$ i.e.
\begin{align*}
\|\mathbf{F}_x^{V_1}-\mathbf{F}_x^{V_2}\|_{\sup}\leq (1-\lambda)\|V_1-V_2\|_{\sup}, \forall V_1,V_2\in\mathcal{V}.
\end{align*}
\end{lemma}
\begin{proof}
Since the second term of mapping $\mathbf{F}_x$ in equation (\ref{eq: F}) is irrelevant to $y$, we have
\begin{align*}
\mathbf{F}_x^{V_1}(p)=&\min_{y\in \Delta(B)}\{\lambda\sum_{k\in K}p^k{x^k}^TM^ky\}+(1-\lambda)\mathbf{T}_{p,x}(V_1),\\
\mathbf{F}_x^{V_2}(p)=&\min_{y\in \Delta(B)}\{\lambda\sum_{k\in K}p^k
{x^k}^TM^ky\}+(1-\lambda)\mathbf{T}_{p,x}(V_2).
\end{align*}
Therefore, according to the definition of $\mathbf{T}$ in (\ref{eq: T}),
\begin{align*}
&|\mathbf{F}_x^{V_1}(p)-\mathbf{F}_x^{V_2}(p)|\\
\leq&(1-\lambda)\sum_{a\in A}\bar{x}_{p,x}(a)|V_1(\pi(a;p,x))-V_2(\pi(a;p,x))|.
\end{align*}
Its supreme norm, hence, satisfies
\begin{align*}
&\|\mathbf{F}_x^{V_1}-\mathbf{F}_x^{V_2}\|_{\sup}\\
\leq&\sup_{p\in \Delta(K)} (1-\lambda)\sum_{a\in A}\bar{x}_{p,x}(a)|V_1(\pi(a;p,x))-V_2(\pi(a;p,x))|\\
\leq & (1-\lambda)\|V_1-V_2\|_{\sup} \sup_{p\in \Delta(K)}\sum_{a\in A}\bar{x}_{p,x}(a)\\
=& (1-\lambda)\|V_1-V_2\|_{\sup}
\end{align*}
\end{proof}

Lemma \ref{lemma: contractor F} further implies that the approximated game value $V_{\lambda,N}$ converges to the real game value $V_\lambda$ exponentially fast with respect to $N$, which is shown in the following theorem.
\begin{theorem}
\label{theorem: regret of value iteration}
Given $\lambda\in (0,1)$, the approximated game value $V_{\lambda,N+1}$ converges to $V_\lambda$ exponentially fast with rate $1-\lambda$, i.e.
\begin{align}
\|V_\lambda-V_{\lambda,N+1}\|_{\sup}\leq & (1-\lambda)\|V_\lambda-V_{\lambda,N}\|_{\sup}\label{eq: value function converges informed}\\
\leq &(1-\lambda)^{N+1} \|V_\lambda\|_{\sup}. \label{eq: value function converges informed 1}
\end{align}
\end{theorem}
\begin{proof}
From equation (\ref{eq: dynamic programming}) and (\ref{eq: value iteration}), we have
\begin{align*}
&|V_\lambda(p)-V_{\lambda,N+1}(p)|\\
=&|\max_{x\in\Delta(A)^{|K|}}\mathbf{F}_{x}^{V_\lambda}(p) - \max_{x\in \Delta(A)^{|K|}}\mathbf{F}_x^{V_{\lambda,N}}(p)|.
\end{align*}
Let $x^*$ and $x^\star$ be the solution to $\max_{x\in\Delta(A)^{|K|}}\mathbf{F}_{x}^{V_\lambda}(p)$ and $\max_{x\in \Delta(A)^{|K|}}\mathbf{F}_x^{V_{\lambda,N}}(p)$, respectively.

Given $p\in \Delta(K)$, if $V_\lambda(p)\geq V_{\lambda,N+1}(p)$, we have
\begin{align*}
|V_\lambda(p)-V_{\lambda,N+1}(p)| &\leq |\mathbf{F}_{x^*}^{V_\lambda}(p)-\mathbf{F}_{x^*}^{V_{\lambda,N}}(p)|\\
&\leq (1-\lambda)\|V_\lambda-V_{\lambda,N}\|_{\sup}.
\end{align*}

Given $p\in \Delta(K)$, if $V_\lambda(p)\leq V_{\lambda,N+1}(p)$, we have
\begin{align*}
|V_\lambda(p)-V_{\lambda,N+1}(p)| &\leq |\mathbf{F}_{x^\star}^{V_\lambda}(p)-\mathbf{F}_{x^\star}^{V_{\lambda,N}}(p)|\\
&\leq (1-\lambda)\|V_\lambda-V_{\lambda,N}\|_{\sup}.
\end{align*}

Therefore, we have for any $p\in \Delta(K)$, $|V_\lambda(p)-V_{\lambda,N+1}(p)| \leq (1-\lambda)\|V_\lambda-V_{\lambda,N}\|_{\sup}$, which further implies equation (\ref{eq: value function converges informed}) and (\ref{eq: value function converges informed 1}).
\end{proof}

In $\lambda$-discounted game $\Gamma_\lambda(p)$, $\bar{\sigma}_{\lambda,N}:K\times \Delta(K)\rightarrow \Delta(A)$ indicates the informed player's stationary strategy that is computed based on the approximated game value $V_{\lambda,N}$, and satisfies the following formula.
\begin{align}
\bar{\sigma}_{\lambda,N}(:,p)=&\argmax_{x\in \Delta(A)^{|K|}} \min_{y\in\Delta(B)}&\left(\lambda \sum_{k\in K} p^k {x^k}^T M^k y\right. \nonumber\\
&\left.+(1-\lambda)\mathbf{T}_{p,x}(V_{\lambda,N})\right), \label{eq: sigma n}
\end{align}
where $\bar{\sigma}_{\lambda,N}(:,p)$ is a $|A|\times|K|$ matrix whose $k$th column is $\bar{\sigma}_{\lambda,N}(k,p)$. Clearly, $\bar{\sigma}_{\lambda,N}(:,p)$ can be also seen as player $1$'s security strategy at stage $1$ in the $\lambda$-discounted $N+1$-stage asymmetric repeated game $\Gamma_{\lambda,N+1}(p)$. Following the same steps as in Theorem \ref{theorem: LP formula for informed one in T-stage game}, we can construct a linear program to compute the approximated game value
$V_{\lambda,N+1}(p)$ and the corresponding approximated security strategy $\bar{\sigma}_{\lambda,N}(k,p)$.

\begin{theorem}
\label{theorem: LP of informed player in discounted game}
Consider a two-player zero-sum $\lambda$-discounted asymmetric game $\Gamma_\lambda(p)$. The approximated game value $V_{\lambda,N+1}(p)$ satisfies
\begin{align}
  V_{\lambda,N+1}(p)=&\max_{q,\ell\in Q,L} \sum_{t=1}^{N+1} \sum_{h_t^A\in H_t^A} \lambda(1-\lambda)^{t-1} \ell_{h_t^A} \label{eq: LP formulationa for informed player in discounted game}\\
s.t. & \sum_{k\in K,a\in A}q_{t+1}(k,(h_t^A,a)) M^k_{a,:}\geq \ell_{h_t^A} \mathbf{1}^T, \nonumber\\
&\forall t=1,2,\ldots,N+1, h_t^A\in H_t^A, \label{eq: LP formulationa for informed player in discounted game 2}
\end{align}
where $q \in Q$ is a set including all properly dimensioned real vectors satisfying (\ref{eq: q constraint 1}-\ref{eq: q constraint 3}), $L$ is a properly dimensioned real space, and $(h_t^A,a_t)$ corresponds to concatenation. The approximated security strategy
\begin{align}
\bar{\sigma}^a_{\lambda,N}(k,p)=q_2^*(a;k), \forall a\in A. \label{eq: informed player's approximated strategy}
\end{align}
\end{theorem}

\subsubsection{The performance analysis of the informed player's approximated security strategy} Now that we can compute the informed player's approximated security strategy, the next question is which performance this strategy can guarantee. To this purpose, we first define the \emph{security level $J^{\bar{\sigma}_{\lambda,N}}(p)$} guaranteed by the approximated security strategy $\bar{\sigma}_{\lambda,N}$ as
\begin{align}
J^{\bar{\sigma}_{\lambda,N}}(p)=\min_{\bar{\tau}\in\bar{\mathcal{T}}}\gamma_{\lambda}(p,\bar{\sigma}_{\lambda,N},\bar{\tau}).
\label{eq: security level}
\end{align}
Since $\bar{\sigma}_{\lambda,N}$ is a stationary strategy, according to the standard procedure of dynamic programming, its security level $J^{\bar{\sigma}_{\lambda,N}}$ has the following property.
\begin{lemma}
\label{lemma: stationary policy stationary dynamic equation}
Let $\bar{\sigma} \in \bar{\Sigma}$ be the informed player's stationary strategy that depends only on the belief state $p_t$ besides the state $k\in K$. The security level $J^{\bar{\sigma}}$ of $\bar{\sigma}$ satisfies $J^{\bar{\sigma}}(p)=\mathbf{F}^{J^{\bar{\sigma}}}_{\bar{\sigma}(:,p)}(p)$.
\end{lemma}
\begin{proof}
Since player 1's strategy is fixed to be $\bar{\sigma}$, the discounted game $\Gamma_\lambda$ becomes a discounted optimization problem, and hence satisfies Bellman's principle, i.e.
\begin{align*}
J^{\bar{\sigma}}(p)
=&\min_{y \in \Delta(B)} \left(\lambda \sum_{k\in K} p^k \bar{\sigma}(k,p)^T M^k y\right.\\
&\left.+ (1-\lambda)\sum_{a\in A}\bar{x}_{p,\bar{\sigma}(:,p)}(a){J^{\bar{\sigma}}(\pi(p,\bar{\sigma}(:,p),a))}\right)\\
=& \min_{y \in \Delta(B)} \left(\lambda \sum_{k\in K} p^k {\bar{\sigma}(k,p)}^T M^k y \right.\\
&\left.+ (1-\lambda)\mathbf{T}_{p,\bar{\sigma}(:,p)}(J^{\bar{\sigma}})\right)\\
=& \mathbf{F}_{\bar{\sigma}(:,p)}^{J^{\bar{\sigma}}}(p).
\end{align*}
\end{proof}

Now, we are ready to show that the difference between the approximated security strategy's security level $J^{\bar{\sigma}_{\lambda,N}}$ and the game value is bounded from above, which is stated in the following theorem.
\begin{theorem}
\label{theorem: error bound on RHP}
The security level $J^{\bar{\sigma}_{\lambda,N}}$ of the informed player's approximated security strategy $\bar{\sigma}_{\lambda,N}$ defined in equation (\ref{eq: sigma n}) converges to the game value exponentially fast in $N$ with rate $1-\lambda$. To be more specific,
\begin{align}
\|V_\lambda-J^{\bar{\sigma}_{\lambda,N}}\|_{\sup}\leq &\frac{2(1-\lambda)}{\lambda}\|V_\lambda-V_{\lambda,N}\|_{\sup}.\label{eq: performance difference of informed player}\\
\leq &(1-\lambda)^{N+1}\frac{2\|V_\lambda\|_{\sup}}{\lambda}.  \label{eq: performance converges informed}
\end{align}
\end{theorem}
\begin{proof}
Lemma \ref{lemma: stationary policy stationary dynamic equation} indicates that
\begin{align*}
&|V_\lambda(p)-J^{\bar{\sigma}_{\lambda,N}}(p)|\\
\leq & |V_\lambda(p)-V_{\lambda,N+1}(p)|+|V_{\lambda,N+1}(p)-J^{\bar{\sigma}_{\lambda,N}}(p)|\\
= & |V_\lambda(p)-V_{\lambda,N+1}(p)|+|\mathbf{F}_{\bar{\sigma}_{\lambda,N}(:,p)}^{V_{\lambda,N}}(p)-\mathbf{F}^{J^{\bar{\sigma}_{\lambda,N}}}_{\bar{\sigma}_{\lambda,N}(:,p)}(p)|.
\end{align*}
Take the supreme norm on both sides, and use Lemma \ref{lemma: contractor F} and Theorem \ref{theorem: regret of value iteration}. We have
\begin{align*}
&\|V_\lambda-J^{\bar{\sigma}_{\lambda,N}}\|_{\sup}\\
\leq&(1-\lambda)\left(\|V_{\lambda}-V_{\lambda,N}\|_{\sup} +\|V_{\lambda,N}-J^{\bar{\sigma}_{\lambda,N}}\|_{\sup}\right)\\
\leq & (1-\lambda)\left(\|V_{\lambda}-V_{\lambda,N}\|_{\sup} +\|V_{\lambda,N}-V_\lambda\|_{\sup}\right.\\
&\left.+\|V_\lambda-J^{\bar{\sigma}_{\lambda,N}}\|_{\sup}\right),
\end{align*}
which implies equation (\ref{eq: performance difference of informed player}). Together with Theorem \ref{theorem: regret of value iteration}, equation (\ref{eq: performance converges informed}) is shown.
\end{proof}
Notice that as $N$ goes to infinity, the difference between the game value and the security level of the approximated security strategy $\bar{\sigma}_{\lambda,N}$ goes to $0$.

We would like to provide an algorithm to conclude this subsection about the informed player's approximated security strategy.
\begin{algorithm}{The informed player's algorithm in $\lambda$-discounted asymmetric repeated game\hfill{}}
  \begin{enumerate}[(i)]
    \item Initialization
        \begin{itemize}
          \item Read payoff matrices $M$, initial probability $p_0$, and system state $k$.
          \item Set $N$.
          \item Let $t=1$, and $p_1=p_0$.
        \end{itemize}
    \item Compute the informed player's approximated security strategy $\bar{\sigma}_{\lambda,N}$ based on (\ref{eq: informed player's approximated strategy}) where $q^*_2$ is the optimal solution of LP (\ref{eq: LP formulationa for informed player in   discounted game}-\ref{eq: LP formulationa for informed player in discounted game 2}) with $p=p_t$.
    \item Choose an action $a\in A$ according to the probability $\bar{\sigma}_{\lambda,N}(k,p_t)$, and announce it publicly.
    \item Update the belief state $p_{t+1}$ according to (\ref{eq: belief state}).
    \item Update $t=t+1$ and go to step (ii).
  \end{enumerate}
\end{algorithm}

\subsection{The uninformed player}
Because of the lack of access to the informed player's strategy, the belief state $p_t$ is not available to the uninformed player, and hence can not serve as the uninformed player's sufficient statistics. De Meyer first introduced the dual game of an asymmetric repeated game in \cite{de1996repeated}, and pointed out that the uninformed player's security strategy in the dual game with a special initial regret is also the uninformed player's security strategy in the `primal' game. One applaudable property of the uninformed player's security strategy in the dual game is that the security strategy depends only on a fixed-sized sufficient statistics that is fully available to the uninformed player. The questions are what is the `special' initial regret, and how to compute the corresponding security strategy in the dual game. To answer these questions, we first introduce the dual game of an asymmetric repeated game.

\subsubsection{The uninformed player's security strategy and the dual game}
Given a $\lambda$-discounted asymmetric repeated game $\Gamma_\lambda(p)$, which is also called the primal game hereafter, its dual game $\tilde{\Gamma}_\lambda(w)$ is defined with respect to $p$, where $w\in\mathbb{R}^{|K|}$ is called the initial regret. The dual game is played the same way as in the primal game, except that the system state $k\in K$ is chosen by player 1 (informed player) instead of the nature. In the dual game, Player 2 (uninformed player) is still not informed of the system state. Let $p$ be player $1$'s strategy to choose the system state, player 1's payoff or player 2's penalty in the dual game $\tilde{\Gamma}_\lambda(w)$ is defined as
\begin{align}
\tilde{\gamma}_\lambda(w,p,\sigma,\tau)=\expect{p,\sigma,\tau}{w^k+\sum_{t=1}^\infty \lambda(1-\lambda)^{t-1}M(k,a_t,b_t)}. \label{eq: tilde gamma}
\end{align}

The $\lambda$-discounted asymmetric repeated dual game $\tilde{\Gamma}_\lambda(w)$ has a game value denoted by $\tilde{V}_\lambda(w)$ satisfying \cite{de1996repeated}
\begin{align}
\tilde{V}_\lambda(w)=\min_{\tau\in\mathcal{T}}\max_{p\in\Delta(K),\sigma\in\Sigma}\tilde{\gamma}_\lambda(w,p,\sigma,\tau) \nonumber \\ =\max_{p\in\Delta(K),\sigma\in\Sigma}\min_{\tau\in\mathcal{T}}\tilde{\gamma}_\lambda(w,p,\sigma,\tau). \label{eq: tilde V lambda}
\end{align}
The game value of the dual game and the game value of the primal game are related in the following way \cite{de1996repeated,sorin2002first}.
\begin{align}
\tilde{V}_\lambda(w)=& \max_{p\in \Delta(K)}\{V_\lambda(p) + p^Tw\},\label{eq: V to V dag}\\
V_\lambda(p)=&\min_{w\in\mathbb{R}^{|K|}}\{\tilde{V}_\lambda(w) - p^Tw\}.\label{eq: V dag to V}
\end{align}

It was shown that the security strategies of the uninformed player in both the primal and the dual games depend only on the informed player's history actions \cite{de1996repeated,sorin2002first}. Following the same steps as in Proposition \ref{prop:noP2also}, we can show that the informed player's best responses to an $H^B$ independent strategy in both the primal and the dual games are also $H^B$ independent. Therefore, we only consider $H^B$ independent strategies for both players in the rest of this subsection.

Let's define the anti-discounted regret $w_t^k$ at stage $t$ with respect to state $k$ given informed player's history action sequence $h_t^A$ as
\begin{align*}
  w_t^k(h_t^A)=\frac{\expect{\bar{\tau}}{w^k+\sum_{s=1}^{t-1}\lambda(1-\lambda)^{s-1}M^k_{a_s,b_s}|k,h_t^A}}{(1-\lambda)^{t-1}},\forall k\in K.
\end{align*}
The anti-discounted regret $w_t^k(h_t^A)$ can be computed recursively as
\begin{align}
w_{t+1}^k((h_t^A,a_t)) =& \frac{w_{t}^k(h_t^A)+\lambda M_{a_t,:}^k \bar{\tau}(h_t^A)}{1-\lambda}, \forall k\in K,  \label{eq: w}
\end{align}
with $w_1=w$.
The anti-discounted regret $w_t$ is indeed the sufficient statistics for the uninformed player in the dual game.
\begin{proposition}\cite{de1996repeated,sorin2002first}
\label{proposition: security strategy relation in primal and dual games}
The value $\tilde{V}_\lambda(w)$ of the $\lambda$-discounted dual asymmetric repeated game $\tilde{\Gamma}_\lambda(w)$ satisfies
\begin{align}
  \tilde{V}_\lambda(w)=&\min_{y\in \Delta(B)}\max_{a\in A}
(1-\lambda)\tilde{V}_\lambda\left(\frac{w+\lambda M_a y}{1-\lambda}\right), \label{eq: recursive formuala for discounted dual game}
\end{align}
where $M_a$ is a $|K|\times |B|$ matrix whose $k$th row is $M_{a,:}^k$.
Moreover, the uninformed player has a security strategy that depends only at each stage $t$ on $w_t$.
\end{proposition}

Meanwhile, it was also shown in \cite{de1996repeated,sorin2002first} that the security strategy of the uninformed player in the dual game $\tilde{\Gamma}_\lambda(w^*)$ with some special initial regret $w^*$ is also the security strategy for the uninformed player in the primal game $\Gamma_\lambda(p)$.
\begin{proposition}{(Corollary 2.10 and 3.25 in \cite{sorin2002first})}
\label{proposition: uninformed player's security strategies in primal and dual games match}
Consider a $\lambda$-discounted asymmetric repeated game $\Gamma_\lambda(p)$ and its dual game $\tilde{\Gamma}_\lambda(w)$. Let $w^*$ be the optimal solution to the optimal problem on the right hand side of equation (\ref{eq: V dag to V}). The security strategy for the uninformed player in the dual game $\tilde{\Gamma}_\lambda(w^*)$ is also the security strategy for the uninformed player in the primal game $\Gamma_\lambda(p)$.
\end{proposition}

Proposition \ref{proposition: uninformed player's security strategies in primal and dual games match} indicates that given an initial probability in the primal game, there exists an initial regret in the dual game such that the security strategies of the uninformed player in the primal and the dual games are the same. Therefore, when playing the primal game $\Gamma_\lambda(p)$, the uninformed player can find out the corresponding initial regret $w^*$ in the dual game first, and then play the dual game instead.

\subsubsection{The special initial regret $w^*$ and its approximation}
The next question is what the special initial regret $w^*$ is. Mathematically speaking, $w^*$ is the optimal solution to the problem $\min_{w\in\mathbb{R}^{|K|}}\{\tilde{V}_\lambda(w) - p^Tw\}$. We are also curious about the physical meaning of $w^*$, i.e. what exactly $w^*$ stands for in the primal game $\Gamma_\lambda(p)$. To this purpose, let's first define the uninformed player's worst case regret $\mu_\lambda\in\mathbb{R}^{|K|}$ of strategy $\bar{\tau}\in\bar{\mathcal{T}}$ in the primal game as
\begin{align}
  \mu_\lambda^k(\tau)=-\max_{\bar{\sigma}(k)\in \bar{\Sigma}(k)}\expect{p,\bar{\sigma},\bar{\tau}}{\sum_{t=1}^\infty \lambda(1-\lambda)^{t-1}M^k_{a_t,b_t}|k},
  \label{eq: mu discounted game}
\end{align}
where $\bar{\sigma}(k)$ indicates informed player's $H^B$ independent behavior strategy if the system state is $k\in K$, and $\bar{\Sigma}(k)$ is the corresponding set including all $\sigma(k)$.

The special initial regret $w^*$ is actually the uninformed player's worst case regret of his security strategy.
\begin{theorem}
  \label{theorem: w^*}
Consider a two-player zero-sum $\lambda$-discounted asymmetric repeated game $\Gamma_\lambda(p)$. Let $\tau^*$ be the uninformed player's security strategy in $\Gamma_\lambda(p)$. An optimal solution $w^*$ to the optimal problem $\min_{w\in\mathbb{R}^{|K|}}\{\tilde{V}_\lambda(w) - p^Tw\}$ is $w^*=\mu_\lambda(\tau^*)$, i.e.
\begin{align}
  \min_{w\in\mathbb{R}^{|K|}}\{\tilde{V}_\lambda(w) - p^Tw\}=\tilde{V}_\lambda(\mu_\lambda(\tau^*)) - p^T\mu_\lambda(\tau^*). \label{eq: optimal solution}
\end{align}
\end{theorem}
\begin{proof}
Equation (\ref{eq: V dag to V}) shows that the left hand side of (\ref{eq: optimal solution}) equals to $V_\lambda(p)$. We will show that the right hand side of (\ref{eq: optimal solution}) equals to $V_\lambda(p)$, too.

First, we have
\begin{align}
V_\lambda(p)=\max_{\sigma\in \Sigma} \gamma_\lambda(p,\sigma,\tau^*)= -p^T\mu_\lambda(\tau^*). \label{eq: 1}
\end{align}

Next, we show that
\begin{align}
\tilde{V}_\lambda(\mu_\lambda(\tau^*))=0. \label{eq: 2}
\end{align}

Equation (\ref{eq: V to V dag}) implies that $\tilde{V}_\lambda(\mu_\lambda(\tau^*))\geq V_\lambda(p)+p^T\mu_\lambda(\tau^*)=0$. Meanwhile, for any $p'\in \Delta(K)$, we have
\begin{align}
  V_\lambda(p')=&\min_{\tau\in\mathcal{T}}\max_{\sigma\in \Sigma} \gamma_\lambda(p',\sigma,\tau)\nonumber\\
  \leq &\max_{\sigma\in \Sigma} \gamma_\lambda(p,\sigma,\tau^*)
  =-p^T\mu_\lambda(\tau^*). \label{eq: 3}
\end{align}
Notice here that $\tau^*$ is the uninformed player's security strategy in $\Gamma_\lambda(p)$, and is not necessarily the uninformed player's security strategy in $\Gamma_\lambda(p')$. Equation (\ref{eq: 3}) implies that $V_\lambda(p)+p^T\mu_\lambda(\tau^*)\leq 0$. Therefore, equation (\ref{eq: 2}) is true. Together with (\ref{eq: 1}), we show that the right hand side of (\ref{eq: optimal solution}) equals to $V_\lambda(p)$, which completes the proof.
\end{proof}

The uninformed player's worst case regret $w^*$ of this security strategy can be seen as the dual variable of the initial probability $p$. The production of the two variables recovers the opposite of the game value (see equation (\ref{eq: 1}). While the informed player's security strategy depends only on $p$ and its Bayesian update $p_t$, the uninformed player can fully rely on $w^*$ and its anti-discounted update $w_t$ to generate his security strategy. Moreover, the belief state $p_t$ is fully available to the informed player, while the anti-discounted regret $w_t$ is fully available to the uninformed player.

Theorem \ref{theorem: w^*} characterizes the physical meaning of the special initial regret $w^*$. The next question is how to compute it. Unfortunately, computing $w^*$ is difficult, since it relies on the security strategy for the uninformed player and the game value in the primal game, which is non-convex \cite{sandholm2010state}. Therefore, we propose to approximate $w^*$ based on the $\lambda$-discounted $N$-stage asymmetric repeated game $\Gamma_{\lambda,N}(p)$, a truncated version of the primal game $\Gamma_\lambda(p)$. Let $\bar{\tau}^\star\in \bar{\mathcal{T}}$ be the security strategy for the uninformed player in $\Gamma_{\lambda,N}(p)$. The approximation $w^\star$ of the special initial regret is $\mu_{\lambda,N}(\bar{\tau}^\star)$ which is defined as $\mu_{\lambda,N}^k(\bar{\tau}^\star)=-\max_{\bar{\sigma}(k)\in \bar{\Sigma}(k)} \expect{p,\bar{\sigma},\bar{\tau}^\star}{\sum_{t=1}^N \lambda(1-\lambda)^{t-1}M^k_{a_t,b_t}|k}$.

Similarly to the $N$-stage game, in the $\lambda$-discounted $N$-stage game, we define the conditional expected total payoff $u_\lambda(\bar{\tau};k,h_{N+1}^A)$ given uninformed player's strategy $\bar{\tau}\in \bar{\mathcal{T}}$, state $k\in K$ and informed player's history action sequence $h_{N+1}^A\in H_{N+1}^A$ as
\begin{align}
  u_{\lambda,N}(\bar{\tau};k,h_{N+1}^A)=\expect{\bar{\tau}}{\sum_{t=1}^N \lambda(1-\lambda)^{t-1}M^k_{a_t,b_t}|k,h_{N+1}^A}, \label{eq: u lambda 1}
\end{align}
which satisfies
\begin{align}
  u_{\lambda,N}(\bar{\tau};k,h_{N+1}^A)=\sum_{t=1}^N \lambda(1-\lambda)^{t-1} M^k_{a_t,:}y_{h_t^A}. \label{eq: u lambda 2}
\end{align}
Following the same steps as in Theorem \ref{theorem: LP formulation of uninformed player in t stage games}, we can construct an LP formulation to compute $V_{\lambda,N}(p)$ and $\mu_{\lambda,N}(\bar{\tau}^\star)$.
\begin{theorem}
  \label{theorem: LP formulation to compute vector payoff}
Consider a $\lambda$-discounted asymmetric repeated game $\Gamma_\lambda(p)$. The approximated game value $V_{\lambda,N}(p)$ satisfies
\begin{align}
  V_{\lambda,N}(p)=&\min_{y\in Y,\ell\in \mathbb{R}^{|K|}}\sum_{k\in K}p^k \ell^k \label{eq: LP initial vector payoff}\\
  s.t.&\ u_{\lambda,N}(y;k,:) \leq \ell^k\mathbf{1}, &&\forall k\in K, \label{eq: LP initial vector payoff 1}\\
  &\ \mathbf{1}^T y_{h_t^A}=1, && \forall h_t^A\in H_t^A,\forall t=1,\ldots,N,\label{eq: LP initial vector payoff 2}\\
  &\ y_{h_t^A}\geq \mathbf{0}, && \forall h_t^A\in H_t^A,\forall t=1,\ldots,N, \label{eq: LP initial vector payoff 3}
\end{align}
where $Y$ is a properly dimensioned real space, and $u_\lambda(y;k,:)$ is a $|H_{N+1}^A|$ dimensional column vector whose element is $u_\lambda(y;k,h_{N+1}^A)$, a linear function of $y$ satisfying equation (\ref{eq:   u lambda 2}). The approximated regret $w^\star$ is $-\ell^{*}$.
\end{theorem}

\subsubsection{The Uninformed player's approximated security strategy}
Now that the approximated initial regret $w^\star$ for the dual game $\tilde{\Gamma}_\lambda(w^\star)$ is computed, the next step is to compute the uninformed player's security strategy in the dual game, which is again non-convex \cite{sandholm2010state}. Similar to what we do in approximating the informed player's security strategy, we use the game value of a $\lambda$-discounted $N$-stage dual game $\tilde{\Gamma}_{\lambda,N}(w^\star)$ to approximate $\tilde{V}_\lambda(w^\star)$, and derive the uninformed player's approximated security strategy based on the approximated game value.

A $\lambda$-discounted $N$-stage asymmetric repeated dual game $\tilde{\Gamma}_{\lambda,N}(w)$ is played the same way as a $\lambda$-discounted asymmetric repeated dual game $\tilde{\Gamma}_\lambda(w)$ except that $\tilde{\Gamma}_{\lambda,N}(w)$ is only played for $N$-stages. Since $\tilde{\Gamma}_{\lambda,N}(w)$ is a finite game, it has a value denoted by $\tilde{V}_{\lambda,N}(w)$, i.e.
\begin{align}
  &\tilde{V}_{\lambda,N}(w)\\
  =&\min_{\bar{\tau}\in\bar{\mathcal{T}}}\max_{p\in \Delta(K),\bar{\sigma}\in \bar{\Sigma}}\expect{p,\bar{\sigma},\bar{\tau}}{w+\sum_{t=1}^T\lambda(1-\lambda)^{t-1}M^k_{a_t,b_t}}\\
  =&\max_{p\in \Delta(K),\bar{\sigma}\in \bar{\Sigma}}\min_{\bar{\tau}\in\bar{\mathcal{T}}}\expect{p,\bar{\sigma},\bar{\tau}}{w+\sum_{t=1}^N\lambda(1-\lambda)^{t-1}M^k_{a_t,b_t} } \label{eq: tilde V lambda T}
\end{align}Following the same steps as in the proof of Proposition 3.23 in \cite{sorin2002first}, we derive that the game value $\tilde{V}_{\lambda,N+1}(w)$ of dual game $\tilde{\Gamma}_{\lambda,N+1}(w)$ satisfies the following recursive formula.
\begin{align}
  \tilde{V}_{\lambda,N+1}(w)=&\min_{y\in \Delta(B)}\max_{a\in A}(1-\lambda)\tilde{V}_{\lambda,N}\left(\frac{w+\lambda M_ay}{1-\lambda}\right),  \label{eq: value iteration dual}
\end{align}
with $\tilde{V}_{\lambda,0}(w)=\max_{k\in K} w^k$.
Moreover, since $\tilde{\Gamma}_{\lambda,N}(w)$ is a dual game of $\Gamma_{\lambda,N}(p)$, their game values have the following relations.
\begin{align}
  \tilde{V}_{\lambda,N}(w)=&\max_{p\in \Delta(K)}\{V_{\lambda,N}(p)+p^Tw\}, \label{eq: game value relation discounted T stage 1}\\
  V_{\lambda,N}(p)=& \min_{w\in\mathbb{R}^{|K|}}\{\tilde{V}_{\lambda,N}(w)-p^Tw\}.\label{eq: game value relation discounted T stage 2}
\end{align}
Based on the relations between the game values of the $\lambda$-discounted game $\Gamma_\lambda(p)$, the $\lambda$-discounted $N$-stage games $\Gamma_{\lambda,N}(p)$ and their dual games, we have the following lemma.
\begin{lemma}
Consider a two-player zero-sum $\lambda$-discounted asymmetric repeated game $\Gamma_\lambda(p)$ and its dual game $\tilde{\Gamma}_\lambda(w)$, and a two-player zero-sum $\lambda$-discounted $N$-stage asymmetric repeated game $\Gamma_{\lambda,N}(p)$ and its dual game $\tilde{\Gamma}_{\lambda,N}(w)$. Their game values satisfy
\begin{align}
  \|V_\lambda-V_{\lambda,N}\|_{\sup}=\|\tilde{V}_\lambda-\tilde{V}_{\lambda,N}\|_{\sup}. \label{eq: game value difference meet}
\end{align}
\end{lemma}
\begin{proof}
First, we show $\|V_\lambda-V_{\lambda,N}\|_{\sup}\leq\|\tilde{V}_\lambda-\tilde{V}_{\lambda,N}\|_{\sup}$. According to equation (\ref{eq: V dag to V}) and (\ref{eq: game value   relation discounted T stage 2}), we have
\begin{align*}
  &|V_\lambda(p)-V_{\lambda,N}(p)|\\
  =&|\min_{w\in\mathbb{R}^{|K|}}\{\tilde{V}_\lambda(w)-p^Tw\}- \min_{w\in\mathbb{R}^{|K|}}\{\tilde{V}_{\lambda,N}(w)-p^Tw\}|.
\end{align*}
Let $w^*$ and $w^\star$ be the optimal solution to the problem $\min_{w\in\mathbb{R}^{|K|}}\{\tilde{V}_\lambda(w)-p^Tw\}$ and $\min_{w\in\mathbb{R}^{|K|}}\{\tilde{V}_{\lambda,N}(w)-p^Tw\}$, respectively. If $\min_{w\in\mathbb{R}^{|K|}}\{\tilde{V}_\lambda(w)-p^Tw\}\geq \min_{w\in\mathbb{R}^{|K|}}\{\tilde{V}_{\lambda,N}(w)-p^Tw\}$, then we have $|V_\lambda(p)-V_{\lambda,N}(p)| \leq |\tilde{V}_\lambda(w^\star)-\tilde{V}_{\lambda,N}(w^\star)|.$ Otherwise, we have $|V_\lambda(p)-V_{\lambda,N}(p)| \leq |\tilde{V}_\lambda(w^*)-\tilde{V}_{\lambda,N}(w^*)|.$ Therefore, for any $p\in \Delta(K)$, $|V_\lambda(p)-V_{\lambda,N}(p)| \leq \|\tilde{V}_\lambda-\tilde{V}_{\lambda,N}\|_{\sup}$, which implies that $\|V_\lambda-V_{\lambda,N}\|_{\sup}\leq\|\tilde{V}_\lambda-\tilde{V}_{\lambda,N}\|_{\sup}$.

Following the same steps, based on equation (\ref{eq: V to V dag}) and (\ref{eq: game value relation   discounted T stage 1}), we derive that $\|\tilde{V}_\lambda-\tilde{V}_{\lambda,N}\|_{\sup}\leq \|V_\lambda-V_{\lambda,N}\|_{\sup}$. Therefore, equation (\ref{eq: game   value difference meet}) is shown.
\end{proof}

Before we draw the uninformed player's approximated security strategy based on the approximated game value $\tilde{V}_{\lambda,N}(w^\star)$, we are interested in how far away the approximated game value is from the real game value. To this purpose, we define an operator $\tilde{\mathbf{F}}_y$ as
\begin{align}
  \tilde{\mathbf{F}}_y^{\tilde{V}}(w)=(1-\lambda)\max_{a\in A}\tilde{V}\left(\frac{w+\lambda M_ay}{1-\lambda}\right), \label{eq: tilde f}
\end{align}
where $y\in \Delta(B)$, $w\in\mathbb{R}^{|K|}$, and $\tilde{V}:\mathbb{R}^{|K|}\rightarrow \mathbb{R}$.
With the same technique as in Lemma \ref{lemma: contractor F}, we can show that $\tilde{\mathbf{F}}$ is also a contraction mapping.
\begin{lemma}
  \label{lemma: contractor tilde f}
Given any $y\in \Delta(B)$ and $\lambda\in (0,1)$, the operator $\tilde{\mathbf{F}}_y$ defined as in (\ref{eq: tilde f}) is a contraction mapping with contraction constant $1-\lambda$, i.e.
\begin{align}
  \|\tilde{\mathbf{F}}_y^{\tilde{V}_1}-\tilde{\mathbf{F}}_y^{\tilde{V}_2}\|_{\sup}\leq (1-\lambda)\|\tilde{V}_1-\tilde{V}_2\|_{\sup},
\end{align}
where $\tilde{V}_{1,2}:\mathbb{R}^{|K|}\rightarrow \mathbb{R}$.
\end{lemma}

Lemma \ref{lemma: contractor tilde f} further implies that the approximated value $\tilde{V}_{\lambda,N}$ converges to the real game value $\tilde{V}_\lambda$ exponentially fast with respect to $N$. The proof is similar to the proof of Theorem \ref{theorem: regret of value iteration}.
\begin{theorem}
  \label{theorem: convergence rate tilde V}
Consider the game value $\tilde{V}_\lambda(w)$ of a $\lambda$-discounted asymmetric repeated dual game $\tilde{\Gamma}_\lambda(w)$ and the game value $\tilde{V}_{\lambda,N}(w)$ of a $\lambda$-discounted $N$-stage asymmetric repeated dual game $\tilde{\Gamma}_{\lambda,N}(w)$. The game value $\tilde{V}_{\lambda,N}$ converges to $\tilde{V}_\lambda$ exponentially fast with respect to the time horizon $N$ with convergence rate $1-\lambda$, i.e.
\begin{align}
  \|\tilde{V}_\lambda-\tilde{V}_{\lambda,N+1}\|_{\sup}\leq (1-\lambda)\|\tilde{V}_\lambda-\tilde{V}_{\lambda,N}\|_{\sup}.
\end{align}
\end{theorem}

Applying the approximated game value $\tilde{V}_{\lambda,N}$ in equation (\ref{eq: recursive formuala for discounted dual game}), we derive the uninformed player's approximated security strategy $\bar{\tau}_{\lambda,N}(w_t)$ in dual game $\tilde{\Gamma}_\lambda(w^\star)$ as
\begin{align}
  \bar{\tau}_{\lambda,N}(w_t)=\argmin_{y\in \Delta(B)} \max_{a\in A} (1-\lambda)\tilde{V}_{\lambda,N}\left(\frac{w_t+\lambda M_ay}{1-\lambda}\right), \label{eq: uninformed player's approximated security strategy}
\end{align}
where $w_t$ is updated according to (\ref{eq: w}). Comparing equation (\ref{eq: uninformed   player's approximated security strategy}) and (\ref{eq: value iteration dual}), we see that the approximated security strategy $\bar{\tau}_{\lambda,N}(w_t)$ can be seen as the uninformed player's security strategy at stage $1$ in a $\lambda$-discounted $N+1$-stage dual game $\tilde{\Gamma}_{\lambda,N+1}(w_t)$. Similar to the LP formulation computing the game value of $\Gamma_{\lambda,N}(p)$, we construct an LP formulation to compute the game value of $\tilde{\Gamma}_{\lambda,N+1}(w)$ and the uninformed player's approximated security strategy $\bar{\tau}_{\lambda,N}$.
\begin{theorem}
\label{theorem: LP formulation uninformed player approximated security strategy}
Consider a two-player zero-sum $\lambda$-discounted $N+1$-stage dual game $\tilde{\Gamma}_{\lambda,N+1}(w)$. Its game value $\tilde{V}_{\lambda,N+1}(w)$ satisfies
\begin{align}
  \tilde{V}_{\lambda,N+1}(w)=&\min_{y\in Y,\ell\in \mathbb{R}^{|K|},L\in\mathbb{R}} L \label{eq: LP uninformed player's approximated security strategy}\\
  s.t.& w+\ell \leq L\mathbf{1} \label{eq: LP uninformed player's approximated security strategy 1}\\
  & \ u_{\lambda,N+1}(y;k,:) \leq \ell^k\mathbf{1}, &&\forall k\in K, \label{eq: LP uninformed player's approximated security strategy 2}\\
  &\ \mathbf{1}^T y_{h_t^A}=1, && \forall h_t^A\in H_t^A,\nonumber\\
  &&&\forall t=1,\ldots,N+1, \label{eq: LP uninformed player's approximated security strategy 3}\\
  &\ y_{h_t^A}\geq \mathbf{0}, && \forall h_t^A\in H_t^A,\nonumber\\
  &&&\forall t=1,\ldots,N+1, \label{eq: LP uninformed player's approximated security strategy 4}
\end{align}
where $Y$ is a properly dimensioned real space, and $u_{\lambda,N+1}(y;k,:)$ is a $|H_{N+2}^A|$ dimensional column vector whose element is $u_{\lambda,N+1}(y;k,h_{N+2}^A)$, a linear function of $y$ satisfying equation (\ref{eq:   u lambda 2}).

Moreover, suppose in dual game $\tilde{\Gamma}_\lambda(w_0)$, at stage $t$, the anti-discounted regret $w_t=w.$ The uninformed player's approximated security strategy $\bar{\tau}_{\lambda,N}(w)$ is $y^*_{h_1^A}.$
\end{theorem}
\begin{proof}
According to equation (\ref{eq: tilde   V lambda T}), we have
\begin{align*}
  \tilde{V}_{\lambda,N+1}(w)=\min_{\bar{\tau}\in\bar{\mathcal{T}}}\max_{p\in \Delta(K)} \sum_{k\in K}p^k(w^k-\mu_{\lambda,N+1}(\bar{\tau})).
\end{align*}
Similar to how we derive equation (\ref{eq: 4}), we have
\begin{align*}
  -\mu_{\lambda,N+1}^k(\bar{\tau})=&\min_{\ell^k\in\mathbb{R}} \ell^k\\
  s.t.& u_{\lambda,N+1}(\bar{\tau};k,:) \leq \ell^k\mathbf{1}.
\end{align*}
Therefore, we have
\begin{align*}
  \tilde{V}_{\lambda,N+1}(w)=&\min_{\bar{\tau}\in\bar{\mathcal{T}}}\max_{p\in \Delta(K)}\min_{\ell\in\mathbb{R}^{|K|}} \sum_{k\in K} p^k(w^k+\ell^k)\\
  s.t.& u_{\lambda,N+1}(\bar{\tau};k,:) \leq \ell^k\mathbf{1}, \forall k\in K.
\end{align*}
Since $\sum_{k\in K}p^k(w^k+\ell^k)$ is bilinear in $p$ and $\ell$, according to Sion's minimax theorem, we have
\begin{align*}
  \tilde{V}_{\lambda,N+1}(w)=&\min_{\bar{\tau}\in\bar{\mathcal{T}}} \min_{\ell\in\mathbb{R}^{|K|}} \max_{p\in \Delta(K)} \sum_{k\in K} p^k(w^k+\ell^k)\\
  s.t.& u_{\lambda,N+1}(\bar{\tau};k,:) \leq \ell^k\mathbf{1}, \forall k\in K.
\end{align*}
According to the duality theorem, given any $\bar{\tau}\in \bar{\mathcal{T}}$ and $\ell\in \mathbb{R}^{|K|}$, we have
\begin{align*}
  &\max_{p\in \Delta(K)} \sum_{k\in K} p^k(w^k+\ell^k)\\
  s.t.& u_{\lambda,N+1}(\bar{\tau};k,:) \leq \ell^k\mathbf{1}, \forall k\in K\\
  =& \min_{L\in \mathbb{R}} L\\
  s.t. & w+\ell \leq L\mathbf{1},\\
  & u_{\lambda,N+1}(\bar{\tau};k,:) \leq \ell^k\mathbf{1}, \forall k\in K,
\end{align*}
which completes the proof.
\end{proof}

Now, we know how to compute the approximated special initial regret $w^\star$ and the uninformed player's approximated security strategy $\bar{\tau}_{\lambda,N}$ in the dual game $\tilde{\Gamma}_\lambda(w^\star)$. This approximated security strategy $\bar{\tau}_{\lambda,N}$ is also the uninformed player's approximated security strategy in the primal game $\Gamma_\lambda(p)$. Let's conclude this subsection with the uninformed player's algorithm in the $\lambda$-discounted asymmetric repeated game $\Gamma_\lambda(p)$.
\begin{algorithm}{The uninformed player's approximated security strategy in $\lambda$-discounted asymmetric repeated game $\Gamma_\lambda(p_0)$ \hfill{}}
\begin{enumerate}[(i)]
  \item Initialization
    \begin{itemize}
      \item Read payoff matrices $M$ and initial probability $p_0$.
      \item Set $N$.
      \item Solve the LP problem (\ref{eq: LP initial vector payoff   1}-\ref{eq: LP initial vector   payoff 2}) with $p=p_0$, and let $w^\star=-\ell^*$.
      \item Let $t=1$ and $w_1=w^\star$.
    \end{itemize}
  \item Solve the LP problem (\ref{eq: LP   uninformed player's approximated security strategy}-\ref{eq: LP uninformed   player's approximated security strategy 4}) with $w=w_t$, and the uninformed player's approximated security strategy $\bar{\tau}(w_t)$ is $y^*_{h_1^A}$.
  \item Choose an action $b\in B$ according to the probability $\bar{\tau}_{\lambda,N}(w_t)$, and announce it publicly.
  \item Read the informed player's action, and update the anti-discounted regret $w_{t+1}$ according to (\ref{eq: w}).
  \item Update $t=t+1$ and go to step (ii).
\end{enumerate}
\end{algorithm}

\subsubsection{The performance difference between the suboptimal strategy and the security strategy}
With the uninformed player's approximated security strategy $\bar{\tau}_{\lambda,N}$, we are interested in the worst case cost guaranteed by this strategy, which is also called the security level of $\bar{\tau}_{\lambda,N}$. Given an uninformed player's strategy $\tau\in \mathcal{T}$, the security level $J^\tau(p)$ in game $\Gamma_\lambda(p)$ is defined as
\begin{align}
  J^\tau(p)=\max_{\sigma\in \Sigma}\gamma_\lambda(p,\sigma,\tau). \label{eq: J tau}
\end{align}
Since the uninformed player's approximated security strategy is derived from his approximated security strategy in the dual game, the security levels of the approximated security strategy in the primal and dual games are highly related. Hence, we would also like to define the security level $\tilde{J}^{\tau}(w)$ of $\tau\in\mathcal{T}$ in the dual game $\tilde{\Gamma}_\lambda(w)$ as
\begin{align}
 \tilde{J}^{\tau}(w)=\max_{p\in \Delta(K)}\max_{\sigma\in \Sigma}\tilde{\gamma}_\lambda(w,p,\sigma,\tau). \label{eq: tilde J tau}
\end{align}
Following the same steps as in the proof of (\ref{eq: V to V dag}-\ref{eq: V dag to V}) in \cite{de1996repeated,sorin2002first}, we can show that $J^\tau(p)$ and $\tilde{J}^\tau(w)$ have the following relations.
\begin{align}
\tilde{J}^\tau(w)=&\max_{p\in \Delta(K)} \{J^\tau(p)+p^Tw\}, \label{eq: to tilde J}\\
  J^\tau(p)=&\min_{w\in\mathbb{R}^{|K|}}\{\tilde{J}^\tau(w)-p^Tw\}. \label{eq: to J}
\end{align}

Meanwhile, we also notice that in dual game $\tilde{\Gamma}_\lambda(w)$, the security level $\tilde{J}^\tau$ of a stationary strategy $\tau$ that depends only on $w_t$ satisfies $\tilde{J}^\tau(w)=\tilde{\mathbf{F}}^{\tilde{J}^\tau}_{\tau(w)}(w)$.
\begin{lemma}
\label{lemma: tilde J recursive}
 Let $\tau\in \mathcal{T}$ be the uninformed player's stationary strategy that depends only on the anti-discounted regret $w_t$. The security level $\tilde{J}^\tau$ of $\tau$ in a $\lambda$-discounted asymmetric information repeated game $\tilde{\Gamma}_\lambda(w)$ satisfies $\tilde{J}^\tau(w)=\tilde{\mathbf{F}}_{\tau(w)}^{\tilde{J}^\tau}(w)$, where $\tilde{\mathbf{F}}_{\tau(w)}$ is defined in (\ref{eq: tilde f}).
\end{lemma}
\begin{proof}
According to Bellman's principle, we have
\begin{align*}
  \tilde{J}^\tau(w)=&\max_{p\in \Delta(k)} \max_{x\in \Delta(A)^{|K|}} \left(\sum_{a\in A,k\in K}p^kx^k(a)w^k\right.\\
  &+\sum_{a\in A, k\in K}\lambda p^kx^k(a)M^k_{a,:}\tau(w)\\
  &\left.+(1-\lambda)\sum_{a\in A}\bar{x}_{p,x}(a)\max_{\sigma\in \Sigma}\gamma_\lambda(\pi(p,x,a),\sigma,\tau)\right)\\
  =&\max_{p\in \Delta(k)}\max_{x\in \Delta(A)^{|K|}} (1-\lambda)\sum_{a\in A} \bar{x}_{p,x}(a) \\
  &\left(\frac{\sum_{k\in K} \pi(p,x,a) (w^k+\lambda M^k_{a,:}\tau(w)) }{1-\lambda} \right.\\
  &\left.+\max_{\sigma\in \Sigma}\gamma_\lambda(\pi(p,x,a),\sigma,\tau)\right)\\
  =&\max_{\bar{x}\in\Delta(A)}(1-\lambda)\sum_{a\in A} \bar{x}(a) \max_{p^+\in\Delta(K)}\max_{\sigma\in \Sigma} \\
  &\left(\frac{\sum_{k\in K} p^{+k} (w^k+\lambda M^k_{a,:}\tau(w)) }{1-\lambda}+\gamma_\lambda(p^+,\sigma,\tau)\right)\\
  =&\max_{a\in A} \tilde{J}^\tau\left(\frac{w+\lambda M_a\tau(w)}{1-\lambda}\right)=\tilde{\mathbf{F}}_{\tau(w)}^{\tilde{J}^\tau}(w)
\end{align*}
\end{proof}

Now, we are ready to analyze the performance difference between the approximated security strategy $\bar{\tau}_{\lambda,N}$ and the security strategy $\bar{\tau}^*$.
\begin{theorem}
  \label{theorem: error bound of player 2's suboptimal strategy}
Consider a two-player zero-sum $\lambda$-discounted asymmetric information repeated game $\Gamma_\lambda(p)$ and the uninformed player's approximated security strategy $\bar{\tau}_{\lambda,N}$ defined in (\ref{eq: uninformed   player's approximated security strategy}).
The security level $J^{\bar{\tau}_{\lambda,N}}(p)$ of $\bar{\tau}_{\lambda,N}$ in game $\Gamma_\lambda(p)$ satisfies
\begin{align}
\|J^{\bar{\tau}_{\lambda,N}}-V_\lambda\|_{\sup}\leq & \frac{2(1-\lambda)}{\lambda}\|V_\lambda-V_{\lambda,N}\|_{\sup}\label{eq: error bound of player 2 in discounted game}
\end{align}
\end{theorem}
\begin{proof}
According to equation (\ref{eq: to J}) and (\ref{eq: V dag to V}), we have $|J^{\bar{\tau}_{\lambda,N}}(p)-V_\lambda(p)|=|\min_{w\in \mathbb{R}^{|K|}} \{\tilde{J}^{\bar{\tau}_{\lambda,N}}(w)-p^Tw\}-\min_{w\in \mathbb{R}^{|K|}}\{\tilde{V}_\lambda(w)-p^Tw\}$. Let $w^*$ be the solution to the optimal problem $\min_{w\in \mathbb{R}^{|K|}}\{\tilde{V}_\lambda(w)-p^Tw\}$. Since $J^{\bar{\tau}_{\lambda,N}}(p)\geq V_\lambda(p)$, we have
\begin{align*}
&|J^{\bar{\tau}_{\lambda,N}}(p)-V_\lambda(p)|\leq |\tilde{J}^{\bar{\tau}_{\lambda,N}}(w^*)-\tilde{V}(w^*)|\\
\leq&\|\tilde{J}^{\bar{\tau}_{\lambda,N}}-\tilde{V}\|_{\sup}, \forall p\in\Delta(K).
\end{align*}

Following the same steps as in the proof of Theorem \ref{theorem: error bound on RHP}, we can show that $\|\tilde{J}^{\bar{\tau}_{\lambda,N}}-\tilde{V}\|_{\sup}\leq \frac{2(1-\lambda)}{\lambda}\|\tilde{V}_{\lambda,N}-\tilde{V}\|_{\sup}.$

Therefore, we have $\|J^{\bar{\tau}_{\lambda,N}}-V_\lambda\|_{\sup}\leq \frac{2(1-\lambda)}{\lambda}\|\tilde{V}_\lambda-\tilde{V}_{\lambda,N}\|_{\sup}.$ According to Equation (\ref{eq: game value difference meet}), equation (\ref{eq: error bound of player 2 in discounted game}) is proved.
\end{proof}

\section{Case Study: A Network Interdiction Problem}
This section uses game theoretic tools to study a network interdiction problem developed from \cite{zheng2012dynamic}, and provides security strategies and approximated security strategies for both players (attacker and network) in finite-horizon game and discounted game, respectively.

Consider a network with a source node and a sink node. There are two channels from the source node to the sink node. One of them has high capacity of $3$, and the other one has low capacity of 1. Only the network knows which channel has high capacity. The network needs to choose a channel to use at each stage to maximize the throughput over a certain horizon. Meanwhile, the attacker will either block one channel with cost 1 or observe the usage of channels with cost $0$ to minimize the throughput over the same horizon. Notice that the attacker can only detect whether a channel is in use, but cannot measure the capacity of a channel. Our objective is to design security or approximated security strategies for both players.

The network interdiction problem is modeled as an asymmetric repeated game with the network to be the informed player and the attacker to be the uninformed player. The network's action is to either use channel 1 (1) or use channel 2 (2), and the attacker's action is to observe (o), block channel 1 (1), or block channel 2 (2). The payoff matrices are provided as in Table \ref{table: payoff matrices}. The initial probability that channel 1 has high capacity is $0.5$.
\begin{table}
  \center
  \caption{Payoff matrix $M^k$ if channel $k$ has high capacity}\label{table: payoff matrices}
  \begin{tabular}{cccc}
    & 1 & 2 & o \\ \cline{2-4}
    1& \multicolumn{1}{|c}{1}& 4 &\multicolumn{1}{c|}{3}\\ \cline{2-4}
    2& \multicolumn{1}{|c}{2} & 1 &\multicolumn{1}{c|}{1} \\ \cline{2-4}
    & \multicolumn{3}{c}{$M^1$}
 \end{tabular}
  \begin{tabular}{cccc}
    & 1 & 2 & o \\ \cline{2-4}
    1& \multicolumn{1}{|c}{1}& 2 &\multicolumn{1}{c|}{1}\\ \cline{2-4}
    2& \multicolumn{1}{|c}{4} & 1 &\multicolumn{1}{c|}{3} \\ \cline{2-4}
    & \multicolumn{3}{c}{$M^2$}
 \end{tabular}
\end{table}

We first compute the security strategies and security levels for both the network and the attacker in a $3$-stage asymmetric game according to Theorem \ref{theorem: LP formula for informed one in T-stage game} and \ref{theorem: LP formulation of uninformed player in t stage games}, respectively. The linear program used to compute the network's security strategy has $65$ constraints and $35$ variables, while the attacker's LP formulation has $44$ constrains and $23$ variables. The security level of the network is $6.57$ which meets the security level of the attacker.

The security strategy of the network is given in Table \ref{table: network's strategy in 3 stage game}. Consider the case in which channel 1 has high capacity. At stage 1, the network uses the high capacity channel with probability $0.64$ instead of 1, because if the network reveals the high capacity channel at stage 1, the attacker will block the high capacity channel for the next two stages. At stage 2, if channel 1 was used at stage 1, then the network thinks that the attacker may guess that channel 1 has high capacity, and hence the network reduces its probability of using channel 1 to $0.56$. Otherwise, the probability of using channel 1 is increased to $0.8$. At the final stage, unless channel 1 is continuously used, the network will use high capacity channel for sure.

The security strategy of the attacker is shown in Table \ref{table: attacker's startegy in 3 stage game}. Notice that because the cost of blocking a channel is low compared with the gain of blocking the high capacity channel, the attacker prefers blocking channels to observing channels. Therefore, for many cases, the attacker launches attacks instead of observing channels unless he is almost sure which channel has high capacity. In this case, because the loss of blocking low capacity channel is higher than the loss of observing channels (see Table \ref{table: payoff matrices}), the attacker would prefer observing channels to blocking low capacity channel. At stage 1, since the initial probability over the states is $[0.5\ 0.5]$, the attacker will block either channel with equal probability. At stage 2, the attacker will increase the probability of blocking channel 1 by $0.04$ if channel 1 is used at stage 1. Otherwise, the probability of blocking channel 1 is decreased by $0.04$. At stage $3$, if one channel was used continuously, the attacker's realized loss in the case that this channel has high capacity is already high, so his strategy focuses more on minimizing the payoff if the continuously used channel has high capacity, as if he is playing only a single game.
\begin{table*}
\center
  \caption{Network's probability of using channel 1}\label{table: network's strategy in 3 stage game}
  \begin{tabular}{||c|c|cc|cccc||}
    \hline
    \backslashbox{channel with high capacity}{$H_t^A$}& $\emptyset$ & 1 & 2 & 11 &12 &21 &22 \\ \hline
    1 & 0.64 &0.56 & 0.8 & 0.4 & 1 & 1& 1 \\ \hline
    2& 0.35 & 0.20 & 0.44 & 0 & 0& 0& 0.6 \\ \hline
  \end{tabular}
  \caption{Attacker's behavior strategy}\label{table: attacker's startegy in 3 stage game}
  \begin{tabular}{||c|c|cc|cccc||}
    \hline
    \backslashbox{Attacker's action}{$H_t^A$}& $\emptyset$ & 1 & 2 & 11 &12 &21 &22 \\ \hline
    1 & 0.5 & 0.54 & 0.46 & 0.68 & 0.49 & 0.51 & 0.04 \\ \hline
    2& 0.5 & 0.46 & 0.54 & 0.04 & 0.51 & 0.49& 0.68 \\ \hline
    o& 0 & 0 & 0& 0.28 & 0 & 0 & 0.28 \\ \hline
  \end{tabular}
\end{table*}

The security strategies of both players are, then, used in the $3$-stage network interdiction game. We ran the $3$-stage game for $5000$ times, and the average total payoff of the network was $6.58$ which was approximately the game value $6.57$ computed according to Theorem \ref{theorem: LP formula for informed one in T-stage game} and \ref{theorem: LP formulation of uninformed player in t stage games}.

\begin{figure}
  \includegraphics[width=.5\textwidth]{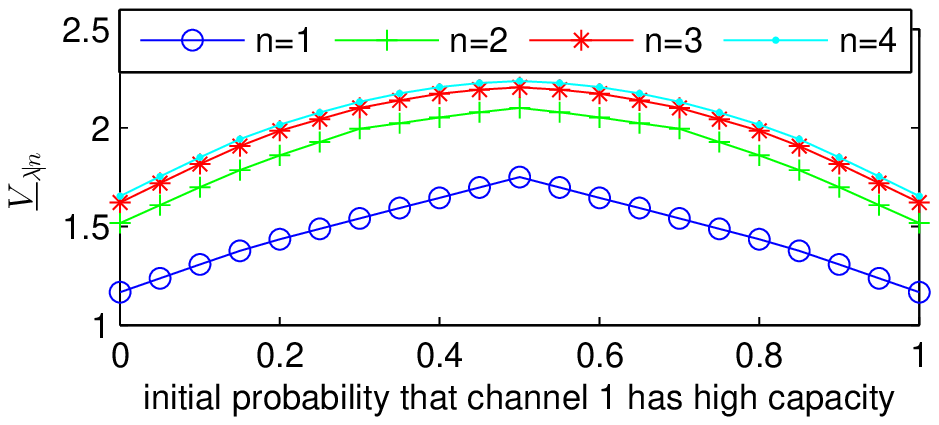}
  \includegraphics[width=.5\textwidth]{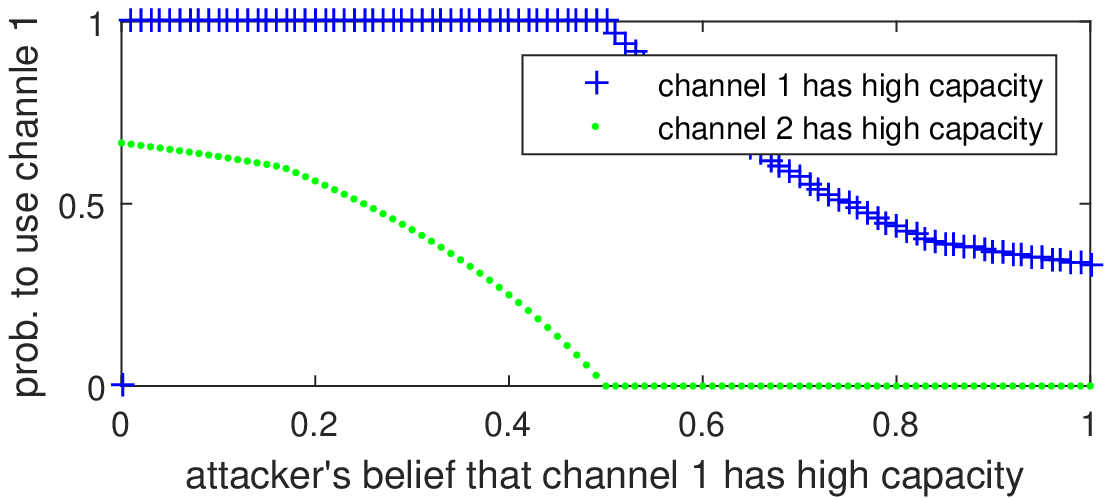}
  \caption{The network's approximated game value and suboptimal strategy in the discounted asymmetric information game}\label{fig: informed player in discounted game}
\end{figure}
Next, we compute the approximated security strategies for both players in a $0.7$-discounted asymmetric repeated game. According to Theorem \ref{theorem: LP of informed player in discounted game}, the network computes his approximated security strategy based on the approximated game value $V_{\lambda,4}$. The game values from the discounted $1$-stage game to the discounted $4$-stage game are presented in the left plot of Figure \ref{fig: informed player in discounted game}. We see that the approximated game value converges, and that the more unsure the attacker is about the high capacity channel, the higher throughput the network can get, and the highest approximated game value is $2.24$ when the initial probability is $[0.5\ 0.5]$. The approximated security strategy is given in the right plot of Figure \ref{fig: informed player in discounted game}. For both cases, the probability of using channel 1 is lower if the network thinks that the attacker has stronger belief that channel 1 has higher capacity. Meanwhile, compared to the case in which channel 2 has high capacity (green dots), it is more possible for the network to use channel 1 if channel 1 has high capacity (blue crosses).

\begin{figure}
  \includegraphics[width=.5\textwidth]{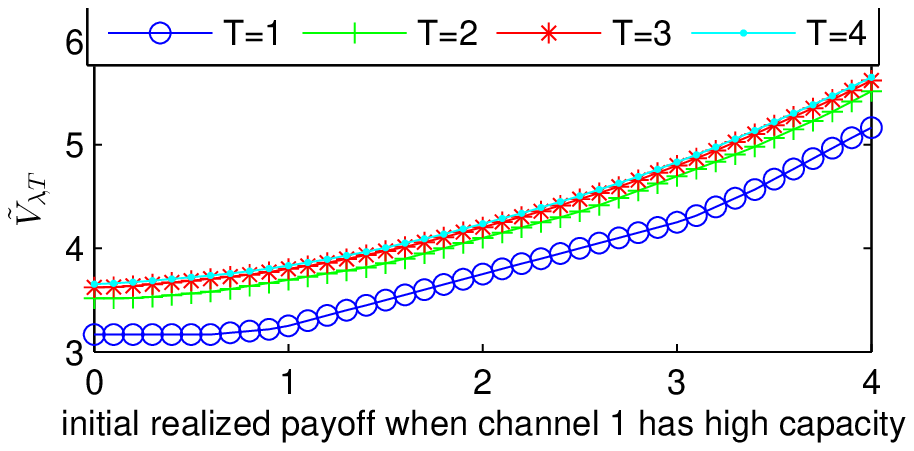}
  \includegraphics[width=.5\textwidth]{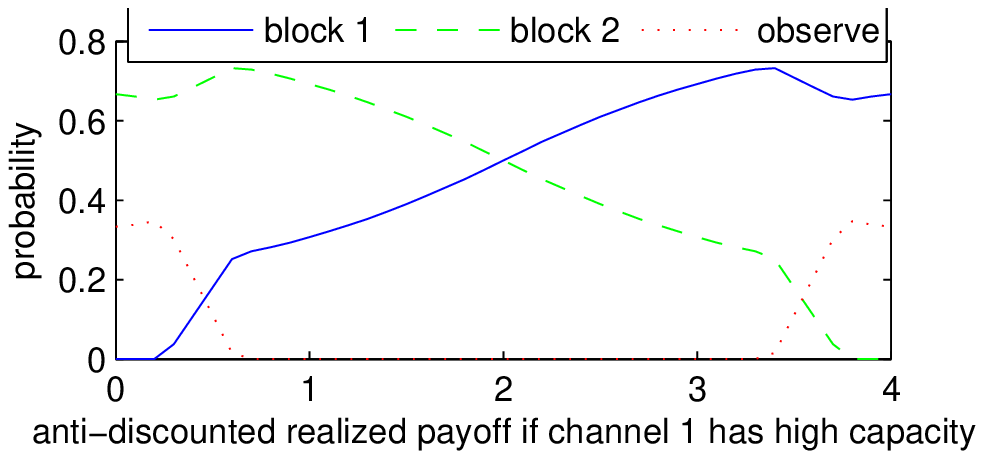}
  \caption{The attacker's approximated game value and suboptimal strategy in the discounted asymmetric information game at stage 2}\label{fig: uninformed player in discounted game 2}
\end{figure}
To compute the attacker's approximated security strategy, we first need to compute the approximated special initial regret $w^\star$ according to Theorem \ref{theorem: LP formulation to compute vector payoff}, which is $[-2.24;-2.24]$ for $N=4$. At each stage, the attacker computes his approximated security strategy based on the anti-discounted regret $w_t$ with $w_1=w^\star$. We assume that the anti-discounted regret $w_t^2$ if channel 2 has high capacity is $2$, and use the approximated value $\tilde{V}_{\lambda,4}(w_t)$ to compute the approximated security strategy. The approximated game values $\tilde{V}_{\lambda,N}$ where $N$ varies from $1$ to $4$ are presented in the left plot of Figure \ref{fig: uninformed player in discounted game 2}. We see that the approximated game value converges over $N$, and increases with respect to $w_t^1$. The attacker's approximated security strategy is shown in the right plot of Figure \ref{fig: uninformed player in discounted game 2}. When $w_t^1$ is relatively low compared with $w_t^2$, the attacker will block channel 2 with higher probability to balance the payoffs of both cases, as if he believes that it is more possible for channel 2 to have high capacity. Contrarily, when $w_t^1$ is larger than $w_t^2$, the attacker will block channel 1 with higher probability to balance the payoffs of both cases, as if he believes that it is more possible for channel 1 to have high capacity.

The approximated security strategies of both players are, then, used in a $0.7$-discounted network interdiction game. Before running the game, we first anticipate the payoff of the game. From equation (\ref{eq: value function converges informed}), we have that $\frac{\|\underline{V}_{\lambda|n}\|_{\sup}^{\Delta(K)}}{1+(1-\lambda)^N}\leq \|V_\lambda\|\leq \frac{\|\underline{V}_{\lambda|n}\|_{\sup}^{\Delta(K)}}{1-(1-\lambda)^N}$. Together with equation (\ref{eq: performance difference of informed player}) and (\ref{eq: value function converges informed}), the network can guarantee a payoff $|J^{\sigma_{\lambda,N}}(p)|\geq (1-\frac{2(1-\lambda)^{N+1}}{\lambda})\frac{\|\underline{V}_{\lambda,N}\|_{\sup}^{\Delta(K)}}{1+(1-\lambda)^N} =1.96$. Together with equation (\ref{eq: error bound of player 2 in discounted game}) and (\ref{eq: value function converges informed}), the attacker can guarantee a payoff $|J^{\tau_{\lambda,N}(p)}|\leq (1+\frac{2(1-\lambda)^{N+1}}{\lambda})\frac{\|\underline{V}_{\lambda|n}\|_{\sup}^{\Delta(K)}}{1-(1-\lambda)^N} =2.59$. Therefore, we anticipate that the payoff is in the interval $[1.96,2.59]$. When running the game, we stopped at stage $10$ since the sum of the payoff after stage $10$ is less than $10^{-4}$. The $10$-stage $0.7$-discounted game was ran for $100$ times, and the average payoff is $2.35$ which is within our anticipated interval, and demonstrates our main results.

\section{Future Work}
This paper studies asymmetric repeated games in which one player has superior information about the game over the other, and provides LP formulations to compute both player's security strategies in finite-horizon games and approximated security strategies in discounted games. In the future, we will generalize these results to the  case in which one player has superior knowledge of one part of the information, while the other player has superior knowledge of the other part.

\bibliography{}

\end{document}